\documentclass[12pt,a4paper]{article}

\usepackage{color,tikz}
\usepackage[unicode,bookmarks,bookmarksopen,bookmarksopenlevel=2,colorlinks,linkcolor=blue,citecolor=green]{hyperref}

\usepackage{amsmath,eucal,amssymb,amsthm,amsfonts}
\usepackage{mathrsfs,graphicx,texdraw}

\newtheorem{theorem}{Theorem}

\def\neprod{\setbox0=\hbox{$\nearrow$}
  \box0\kern-1.6em\prod} 
\def\swprod{\setbox0=\hbox{$\swarrow$}%
  \,\,\raise.03em\box0\kern-1.18em\prod} 
\def\seprod{\setbox0=\hbox{$\searrow$}%
  \,\,\raise.00em\box0\kern-2.0em\prod} 

\def\im{{\mbox{Im}}}

\def\openone{\leavevmode\hbox{\small1\kern-3.3pt\normalsize1}}

\def\re{\mathrm{Re\,}}

\def\bbbe{{\Bbb E}}
\def\bbbc{{\Bbb C}}
\def\bbbr{{\Bbb R}}

\def\p{{\boldsymbol p}}
\def\q{{\boldsymbol q}}

\begin{document}

\begin{center}
{\LARGE \bf Real Hamiltonian forms of affine Toda field \\[8pt]
theories: spectral aspects}

\bigskip

{\bf Vladimir S. Gerdjikov$^{a,b,}$\footnote{E-mail: {\tt gerjikov@inrne.bas.bg}},
Georgi  G. Grahovski$^{c,}$\footnote{E-mail: {\tt grah@essex.ac.uk}} and
Alexander A. Stefanov$^{a,e,}$\footnote{E-mail: {\tt aleksander.a.stefanov@gmail.com}
}}

\end{center}

\medskip

\noindent
{\it $^{a}$ Institute of Mathematics and Informatics, Bulgarian Academy of Sciences,
 8 Acad. G. Bonchev str.,  1113 Sofia, Bulgaria }\\[5pt]
{\it $^{b}$Institute for Advanced Physical Studies, 111 Tsarigradsko chaussee,
Sofia 1784, Bulgaria}\\[5pt]
{\it $^{c}$ Department of Mathematical Sciences, University of Essex, Wivenhoe Park, Colchester, UK}\\[5pt]
{\sl $^d$ Faculty of Mathematics and Informatics, Sofia University "St. Kliment Ohridski"
 5 James Bourchier Blvd.,  1164 Sofia, Bulgaria } \\[5pt]

\begin{abstract}
\noindent
The paper is devoted to real Hamiltonian forms of 2-dimensional Toda field theories
related to  exceptional simple Lie algebras, and to the spectral theory of the associated Lax operators.
Real Hamiltonian forms are a special type of ``reductions'' of Hamiltonian systems, similar to  real forms of
 semi-simple Lie algebras. Examples of real Hamiltonian forms of affine Toda field theories related to exceptional
 complex untwisted affine Kac-Moody algebras are studied.
Along with the associated Lax representations, we also
formulate the relevant Riemann-Hilbert problems and derive the minimal sets of scattering
data that determine uniquely the scattering matrices and the potentials of the Lax operators.

\end{abstract}


\section{Introduction}\label{sec:intro}

Affine Toda field theories \cite{OlPerMikh,BowCor2,BowCor3,Olive,Olove2} received considerable attention of the Mathematical physics community in the past three decades and
are one of the best understood integrable massive field theories at classical and quantum levels in $1+1$ dimensions \cite{BBT}. The interest in Toda field theories was inspired by  the work of A. Zamolodchikov \cite{Zamol} on deformation of conformal field theories  preserving
integrabilty. It was shown in \cite{Zamol} that the resultant theory is characterised by eight masses
related to the Cartan matrix of $E_8$, and by integrals of motion with spins given by the
exponents of $E_8$ modulo its Coxeter number.

Affine Toda field theories (ATFT) are integrable models of real scalar field ${\bf q}=(q_1, \dots q_n)$, in one space
dimension  with exponential interactions. The Lagrangian of the theory is given by
\begin{equation}\label{eq: ATFT-Lagr}
{\cal L}[{\bf q}] = {1\over 2}(\partial_\mu {\bf q} \cdot \partial^\mu {\bf q})- {m^2\over \beta^2}\sum_{k=0}^{n} n_k \left({\rm e}^{\beta (\alpha_k \cdot{\bf q})} -1\right),
\end{equation}
where the field ${\bf q}(x,t)$ is an $n$-dimensional vector. This corresponds to equations of motion of the form
\begin{equation}\label{eq:ATFT}
{\partial ^2 {\bf q}  \over \partial x \partial t } =
\sum_{j=0}^{r} n_j \alpha_j e^{-(\alpha_j ,{\bf q}(x,t))},
\end{equation}
Each ATFT is associated to a (finite) simple Lie algebra ${\frak g}$ \cite{AMV,BBT}. Here $n$ is the rank of ${\frak g}$, $\alpha_k$s ($k = 1,\dots,n$) are the simple roots of ${\frak g}$ and $\alpha_0$ is the minimal root.
Thus,  \eqref{eq: ATFT-Lagr} are described by the extended Dynkin diagrams of the associated affine
algebra $\hat{\frak g}$. The fields $q_k$ can be rescaled so that $\beta$ only appears in ${\cal L}$ through a common factor $1/\beta^2$ (expanding in powers of $\beta^2$ is equivalent in  Quantum theory to expanding in $\hbar$) \cite{BowCor2,BowCor3}.

Non-simply laced theories are obtained from the affine Toda theories based on simply
laced algebras by folding the corresponding Dynkin diagrams. The same process, called
classical ``reduction'', provides solutions of a non-simply laced theory from the classical
solutions with special symmetries of the parent simply laced theory.

One of the important features of affine Toda field theories is their integrability by the inverse scattering method (ISM) \cite{FaTa,ZMNP}. The starting point is the existence of the so-called Lax operator (see \eqref{eq:2.1} below). The interpretation of the ISM as a generalized Fourier transforms \cite{AKNS,IP2,GeYa} allows one to study all the
fundamental properties of the corresponding  nonlinear evolutionary equations (NLEE's): i)~the description of the class of NLEE related to a~given Lax operator
$L(\lambda)$ and solvable by the ISM; ii)~derivation of the infinite family of integrals of motion;  iii)~their hierarchy
of Hamiltonian structures \cite{GVYa}; and iv) description  of the gauge equivalent systems \cite{GGMV,Grah,Grah2}.

Real Hamiltonian forms (RHF) are another type of ``reductions'' of Hamiltonian systems.
The extraction of RHFs is similar to the obtaining a real forms of
a semi-simple Lie algebra. The Killing form for the later is indefinite in general (it is negatively-
definite for the compact real forms). So one should not be surprised of getting RHF's with
indefinite kinetic energy quadratic form. Of course this is an obstacle for their quantization.

The purpose of this paper is to outline the spectral theory of the Lax operators of real Hamiltonian forms of affine Toda field theories related to complex untwisted exceptional affine Lie algebras.

The structure of the paper is as follows: In Section 2 we provide a brief summary of the Lax representation of ATFTs and all necessary Lie algebraic background knowledge. As to the root systems of the exceptional algebras we are following the conventions in \cite{Bourb1}. In Section 3 first we describe the general method for constructing RHFs. Then we briefly describe the sets of admissible roots of the exceptional Lie algebras, and the explicit constructions of their RHFs.
Section 4 is devoted to the spectral theory of the Lax operators of ATFTs with ${\Bbb Z}_h$-reductions, where $h$ is the Coxeter number of ${\frak g}$. We first provide the general construction of the fundamental analytic solutions $\xi_\nu (x,t,\lambda)$ (FAS) of the Lax operator and show their relevance to the Riemann-Hilbert problem (RHP) on a set of $2h_g$ rays $l_\nu$ closing angles $\pi /h_g$. Then we introduce the asymptotics of the FAS  $\xi_\nu (x,t,\lambda)$ for $x\to \pm \infty$ and $\lambda \in l_\nu$ which determine the scattering data of $L$. This Section ends with a theorem specifying the minimal sets of scattering data for generic choice of the simple Lie algebra $\mathfrak{g}$. Section 5 contains the specific data, concerning each of the exceptional Lie algebras that allows one to determine the minimal set os scattering data for each case. The paper ends up with Conclusions.

\section{Affine Toda field theories: preliminaries}\label{sec:prelim}

To each simple Lie algebra $\mathfrak{g} $ one can relate a Toda field
theory  in $1+1 $ dimensions. It allows a Lax representation in a zero-curvature form:
\begin{equation}\label{eq:0.1}
[L,M]=0,
\end{equation}
where $L $ and $M $ are first order ordinary differential (Lax)
operators:
\begin{eqnarray}\label{eq:2.1}
L\psi \equiv \left(  i{d  \over dx } - iq_x(x,t) - \lambda
J_0\right)
\psi (x,t,\lambda )=0, \\
M\psi \equiv \left(  i{d  \over dt } -  {1\over \lambda}
I(x,t)\right) \psi (x,t,\lambda )=0.
\end{eqnarray}
whose potentials take values in $\mathfrak{g} $.
Here also $q(x,t) \in \mathfrak{h}$ - the Cartan subalgebra of
$\mathfrak{g}$, $\q(x,t)=(q_1,\dots , q_r) $ is its dual $r
$-component vector  ($r=\mbox{rank}\,\mathfrak{ g}$). The potentials of the Lax operators are chosen as follows
\begin{equation}\label{eq:2.2}
J_0 = \sum_{\alpha \in \pi}^{} E_{\alpha },\qquad I(x,t) =
\sum_{\alpha \in \pi}^{} e^{-(\alpha ,\q(x,t))} E_{-\alpha }.
\end{equation}
Here $\pi_{\mathfrak{g}} $ stands for the set of admissible roots of
$\mathfrak{g} $, i.e.  $\pi_{\mathfrak{g}} = \{\alpha _0, \alpha
_1,\dots, \alpha _r\} $, with $\alpha _1,\dots, \alpha _r $ being
the simple roots of $\mathfrak{ g}$ and $\alpha _0 $ being the
minimal root of $\mathfrak{ g} $.  The corresponding Toda field theory is known
as  affine Toda field theory (ATFT). The Dynkin graph corresponding to the set of
admissible roots $\pi_{\mathfrak{g}}=\{\alpha _0, \alpha _1,\dots,\alpha _r\} $ of ${\frak g}$ is called extended Dynkin diagram (EDD). The equations of motion are of the form:
\begin{equation}\label{eq:2.3}
{\partial ^2 \q  \over \partial x \partial t } =
\sum_{j=0}^{r} n_j \alpha_j e^{-(\alpha_j ,\q(x,t))},
\end{equation}
where $n_j $ are the minimal positive  integer coefficients $n_k $ that
provide the decomposition of the $\alpha _0 $ over the simple roots of
$\mathfrak{g} $:
\begin{equation}\label{eq:n_k}
-\alpha _0 = \sum_{k=1}^{r} n_k\alpha _k.
\end{equation}
It is  well known that ATFT models are an infinite-dimensional Hamiltonian system.  The (canonical)
Hamiltonian structure is given by:
\begin{eqnarray}\label{eq:H-g}
H_{\mathfrak{g}} &=& \int_{-\infty }^{\infty }dx\, \mathcal{H}_{\mathfrak{g}}(x,t), \qquad
\mathcal{H}_{\mathfrak{g}}(x,t)= { 1\over 2} ({\bf p}(x,t),
{\bf p}(x,t))+\sum_{k=0}^{r} n_k(e^{-({\bf q}(x,t),\alpha _k)}-1)  ,\\
\label{eq:ome-g}
\Omega _{\mathfrak{g}} &=& \int_{-\infty }^{\infty } dx\,
\omega_{\mathfrak{g}} (x,t), \qquad
\omega_{\mathfrak{g}} (x,t)=( \delta {\bf p}(x,t)\wedge \delta {\bf q}(x,t)),
\end{eqnarray}
where $H_{\mathfrak{g}} $ is the canonical Hamilton function and $\Omega _{\mathfrak{g}}$ is the canonical symplectic structure.
Here also ${\bf p} = d{\bf q}/dt $ are the canonical momenta
and coordinates satisfying canonical Poisson brackets:
\begin{equation}\label{eq:ATFT-PB}
\{ q_k(x,t) , p_j(y,t)\} = \delta_{jk} \delta (x-y).
\end{equation}
The infinite-dimensional phase space $\mathcal{M}=\{{\bf q}(x,t),{\bf p}(x,t)\}$ is spanned by the canonical coordinates and momenta.

\section{Real Hamiltonian forms and affine Toda field theories}\label{sec:RHF}

The Lax representations of the ATFT models widely discussed in the
literature (see e.g. \cite{Mikh,OlPerMikh,Olive,SasKha} and the
references therein) are related mostly to the normal real form of the Lie
algebra $\mathfrak{ g} $, see \cite{Helg}. Here we will study real Hamiltonian forms of ATFT models. The notion of real Hamiltonian forms was introduced in \cite{2} and used to study reductions of ATFTs in \cite{GG}. After a brief outline of the basic theory (following \cite{2}), in this Section we will describe the real Hamiltonian forms for ATFT models related to exceptional untwisted complex Kac-Moody algebras \cite{Kac,VinOni,Xu}.

\subsection{Real Hamiltonian forms}\label{ssec:rHF}

The starting point in the construction of real Hamiltonian forms (RHF) is the complexification of the field involved in the Hamiltonian of the model. First for dynamical variables we consider complex-valued fields
\[
{\bf q}^\bbbc = {\bf q}^0 + i{\bf q}^1, \qquad {\bf p}^\bbbc =
{\bf p}^0 + i{\bf p}^1.
\]
Next we introduce an involution $\mathcal{ C} $ acting on the
phase space $\mathcal{ M} \equiv \{q_k(x), p_k(x)\}_{k=1}^{n} $ as
follows:
\begin{eqnarray}\label{eq:Cc}
&& \mbox{1)} \qquad \mathcal{ C}(F(p_k,q_k)) = F(\mathcal{ C}(p_k),
\mathcal{ C}(q_k)),  \nonumber\\
&& \mbox{2)} \qquad \mathcal{ C}\left( \{ F(p_k,q_k), G(p_k,q_k)\}\right) =
\left\{ \mathcal{ C}(F), \mathcal{ C}(G) \right\} , \\
&& \mbox{3)} \qquad \mathcal{ C}(H( p_k,q_k)) = H(p_k,q_k) . \nonumber
\end{eqnarray}
Here $F(p_k,q_k) $, $G(p_k,q_k) $ and the Hamiltonian $H(p_k,q_k) $ are
functionals on $\mathcal{M} $ depending analytically on the fields
$q_k(x,t) $ and $p_k(x,t) $.

The complexification of the ATFT is rather straightforward. The resulting
complex ATFT (CATFT) can be written down as standard Hamiltonian system with
twice as many fields ${\bf q}^a(x,t) $, ${\bf p}^a(x,t)  $, $a=0,1 $:
\begin{equation}\label{eq:qp-c}
{\bf p}^\bbbc (x,t) = {\bf p}^0(x,t)+i {\bf p}^1(x,t), \qquad
{\bf q}^\bbbc (x,t)= {\bf q}^0(x,t)+i {\bf q}^1(x,t),
\end{equation}
\begin{equation}\label{eq:qp-pb}
\{{q}_{k}^0(x,t), {p}_{j}^0(y,t) \}= - \{{q}_{k}^1(x,t),
{p}_{j}^1(y,t) \} = \delta _{kj} \delta (x-y).
\end{equation}
The densities of the corresponding Hamiltonian and symplectic form
equal
\begin{eqnarray}\label{eq:H_0}
\mathcal{H}_{\rm ATFT}^\bbbc &\equiv & \re \mathcal{H}_{\rm ATFT}
({\bf p}^0+i {\bf p}^1, {\bf q}^0+i {\bf q}^1) \nonumber\\
&=&  {1\over 2 } ({\bf p}^0,{\bf p}^0) -{1\over 2 }
({\bf p}^1,{\bf p}^1) + \sum_{k=0}^{r} e^{-({\bf q}^0,\alpha _k)}
\cos (({\bf q}^1,\alpha _k)) ,  \\
\label{eq:ome_0}
\omega^\bbbc &=& (d{\bf p}^0\wedge  i d{\bf q}^0) - (d{\bf p}^1\wedge
d {\bf q}^1).
\end{eqnarray}
The family of RHF then are obtained from the CATFT by imposing an
invariance condition with respect to the involution
$\tilde{\mathcal{ C}} \equiv \mathcal{ C}\circ \ast $ where by
$\ast $ we denote the complex conjugation. The involution
$\tilde{\mathcal{ C}} $ splits the phase space $\mathcal{ M}^\bbbc
$ into a direct sum $\mathcal{ M}^\bbbc \equiv {\cal M}_+^\bbbc
\oplus \mathcal{M}_-^\bbbc$ where
\begin{equation}\label{eq:M-c}\begin{aligned}
\mathcal{M}_+^\bbbc &= \mathcal{ M}_0 \oplus i \mathcal{ M}_1, &\quad \mathcal{M}_-^\bbbc &= \mathcal{ M}_0 \oplus -i \mathcal{ M}_1, \\
 \mathcal{C}(\q^+ + i \q^-) &= (\q^+ - i \q^-), &\quad  \mathcal{C}(\p^+ + i \p^-) &= (\p^+ - i \p^-).
\end{aligned}\end{equation}
Each involution $\mathcal{C} $ induces an involution
(involutive automorphism) $\mathcal{C}^\# $ also of $\mathfrak{g} $.
Below we will choose $\mathcal{C}^\# $ in a special way, which is related to the Coxeter
automorphism defined below in its dihedral form (\ref{eq:DCox-E6}), see \cite{Bourb1, Helg, Cart}.
Indeed, we will define $\mathcal{C}^\# $ using the reflection $S_1$ with respect to the set of
black roots of the Dinkin dyagram (see Figure \ref{fig:Dynkin-E6}). More precisely:
\begin{equation}\label{eq:alj}\begin{aligned}
\mathcal{C}^\# \alpha_j &=  \alpha_j \quad \mbox{for}\quad  \alpha_j \in W_g, &\qquad
\mathcal{C}^\# \beta_j &=  -\beta_j \quad \mbox{for} \quad \beta_j \in B_g , \\
{\bf p}^0(x,t) &= \sum_{\alpha_j\in W_g} {\bf p}_j(x,t) \alpha_j,    &\qquad   {\bf p}^1(x,t) &= \sum_{\beta_j \in B_g} {\bf p}_j(x,t) \beta_j, \\
{\bf q}^0(x,t) &= \sum_{\alpha_j\in W_g} {\bf q}_j(x,t) \alpha_j,    &\qquad   {\bf q}^1(x,t) &= \sum_{\beta_j \in B_g} {\bf q}_j(x,t) \beta_j,
\end{aligned}\end{equation}
where $W_g$ (resp. $B_g$) is the set of white roots (resp. black roots) of the Dinkin diagram.

Thus to each involution $\mathcal{ C} $ one can relate a RHF of the ATFT.  Note that $\mathcal{C}^\# $
preserves the system of admissible roots of $\mathfrak{g} $ and the
extended Dynkin diagrams of $\mathfrak{ g} $ studied in \cite{SasKha}.
 For a sake of brevity below we will skip the dependence of ${\bf p} $ and ${\bf q} $ on
$x $ and $t $. Indeed, the condition 3) in \eqref{eq:Cc} requires that:
\begin{equation}\label{eq:*1}
(\mathcal{C}(\q ),\alpha ) = ({\bf q} , \mathcal{C}^\# (\alpha )), \qquad \alpha \in \pi_{\mathfrak{g}},
\end{equation}
and therefore we must have $\mathcal{C}(\pi_{\mathfrak{g}}) = \pi_{\mathfrak{g}} $.
As a result:
\begin{equation}\label{eq:palj}\begin{aligned}
\mathcal{C}( {\bf p}^0(x,t), \alpha_j) &= ({\bf p}^0(x,t),\alpha_j), &\qquad \mathcal{C}( {\bf p}^1(x,t), \beta_j) &= -({\bf p}^1(x,t),\beta_j), \\
( {\bf p}^0(x,t), \beta_j) &= 0, &\qquad ( {\bf p}^1(x,t), \alpha_j) &= 0, \\
\mathcal{C}( {\bf q}^0(x,t), \alpha_j) &= ({\bf q}^0(x,t),\alpha_j), &\qquad \mathcal{C}( {\bf q}^1(x,t), \beta_j) &= -({\bf q}^1(x,t),\beta_j), \\
( {\bf q}^0(x,t), \beta_j) &= 0, &\qquad ( {\bf q}^1(x,t), \alpha_j) &= 0,
\end{aligned}\end{equation}
where $\alpha_j \in W_g$, and $\beta_j \in B_g$.
Then applying the ideas of \cite{Mikh} we obtain the following result
for the RHF of the ATFT related to $\mathfrak{g} $:
\begin{equation}\label{eq:*H-g}\begin{aligned}
\mathcal{H}^{\bbbr}_{\mathfrak{g}} &= {1 \over 2} ({\bf p}^+ , {\bf p}^+ ) -{1 \over 2} ({\bf p}^- , {\bf p}^- ) +
\sum_{\alpha_k\in W_g} n_k' (e^{-({\bf q}^+ ,\alpha _k)} -1) + \sum_{\beta_k\in B_g} n_k'' ( \cos ({\bf q}^- ,\beta_k)-1),
\end{aligned}\end{equation}
\begin{equation}\label{eq:*ome-g}\begin{split}
\omega ^{\bbbr}_{\mathfrak{g}} = (\delta {\bf p}^+  \wedge \delta {\bf q}^+) -(\delta {\bf p}^-  \wedge \delta {\bf q}^-) ,
\end{split}\end{equation}
The Hamiltonian along with the terms related to the simple roots, contains also the minimal root $-\alpha_0$.
It is well known that the maximal root $\alpha_0$ can be expanded as sum of the simple roots with integer nonnegative coefficients:
\begin{equation}\label{eq:a0}\begin{split}
 \alpha_0 = \sum_{j=1}^{r} n_j \alpha_j = \sum_{\alpha_k\in W_g} n_k' \alpha_k + \sum_{\beta_j\in B_g} n_k'' \beta_j .
\end{split}\end{equation}
In the second expression above we have separated the terms with the white and black roots.

The RHF of ATFT are more general integrable systems than the models
described in \cite{Evans, Evans2, Evans3, SasKha} which involve only the fields ${\bf q}^+ $, ${\bf p}^+
$ invariant with respect to $\mathcal{C} $.

\subsection{Affine Toda field theories related to $E_6^{(1)}$}\label{ssec:E6-1}

The set of admissible roots for this algebra is
\begin{equation}\label{eq:e6roots}\begin{aligned}
\alpha_1&={1\over 2}(e_1-e_2-e_3-e_4-e_5-e_6-e_7+e_8), &\quad \alpha_2 &=e_1+e_2, \\
\alpha_3 &=e_2-e_1,\qquad \alpha_4=e_3-e_2, \qquad \alpha_5=e_4-e_3, &\quad \alpha_6 &=e_5-e_4,\\
\alpha_0 &=-{1\over 2}(e_1+e_2+e_3+e_4+e_5-e_6-e_7+e_8).
\end{aligned}\end{equation}
Here $\alpha_1,\dots , \alpha_6$ form the set of simple roots of
$ E_6$ and $\alpha_0$ is the minimal root of the algebra. The extended Dynkin diagram of $E_6^{(1)}$ is shown on Figure \ref{fig:Dynkin-E6}.
This is the standard definition of the root system of $E_6
$ embedded into the 8-dimensional Euclidean space $\bbbe^8 $. The
root space $\bbbe_6 $ of the algebra $E_6$ is the
6-dimensional subspace of $\bbbe^8 $ orthogonal to the vectors
$e_7+e_8 $ and $e_6+e_7+2e_8 $. Thus any vector ${\bf q} $
belonging to $\bbbe_6 $ has only 6 independent coordinates and can
be written as:
\begin{gather}\label{eq:vec-q}
{\bf q} = \sum_{k=1}^{5} q_k e_k + q_6 e'_6, \qquad e'_6 = {1\over
\sqrt{3} } (e_6 + e_7 - e_8).
\end{gather}
The fundamental weights of $E_6$ are
\begin{eqnarray*}
\omega_1&=&{2\over 3} (e_8-e_7-e_6), \qquad \omega_2= {1\over 2}(e_1+e_2+e_3+e_4+e_5-e_6-e_7+e_8);\\
\omega_3 &=& {1\over 2}(-e_1+e_2+e_3+e_4+e_5)+{5\over 6}(e_8-e_7-e_6), \qquad
\omega_4 = e_3+e_4+e_5-e_6-e_7+e_8,\\
\omega_5 &=& e_4+e_5 + {2\over 3} (e_8-e_7-e_6), \qquad \omega_6=e_5+ {1\over 3}(e_8-e_7-e_6).
\end{eqnarray*}

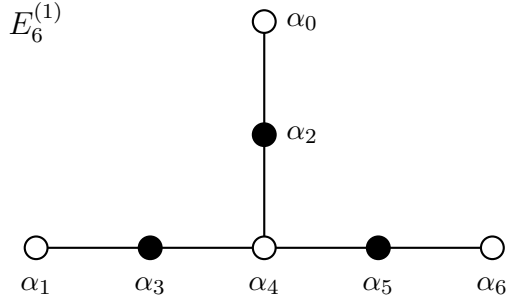
\begin{figure}
\begin{center}
\begin{tikzpicture}
\draw[thick] (0,0) --(6,0);
\draw[thick] (3,0) --(3,3);
\draw[thick,fill=white] (0,0) circle (1.5 mm);
\draw[thick,fill=black] (1.5,0) circle (1.5 mm);
\draw[thick,fill=white] (3,0) circle (1.5 mm);
\draw[thick,fill=black] (4.5,0) circle (1.5 mm);
\draw[thick,fill=white] (6,0) circle (1.5 mm);
\draw[thick,fill=black] (3,1.5) circle (1.5 mm);
\draw[thick,fill=white] (3,3) circle (1.5 mm);
\draw (0,-0.5) node {\small $\alpha_1$};
\draw (1.5,-0.5) node {\small $\alpha_3$};
\draw (3,-0.5) node {\small $\alpha_4$};
\draw (4.5,-0.5) node {\small $\alpha_5$};
\draw (6,-0.5) node {\small $\alpha_6$};
\draw (3.5,1.5) node {\small $\alpha_2$};
\draw (3.5,3) node {\small $\alpha_0$};
\draw (0.0,3) node {$E_6^{(1)}$};

\end{tikzpicture}
\end{center}
  \caption{\small The extended Dynkin diagram of the  complex untwisted affine Kac-Moody algebra $E_6^{(1)}$. The white roots are invariant with respect to the automorphism $S_2$ in \eqref{eq:DCox-E6}.}\label{fig:Dynkin-E6}
\end{figure}

\noindent
We will use the dihedral realization of the Coxeter automorphism for $E_6^{(1)}$:
\begin{equation}\label{eq:DCox-E6}
{\rm Cox}\, (E_6^{(1)})=S_1\circ S_2, \qquad S_1=  S_{\alpha_2}\circ S_{\alpha_3}\circ S_{\alpha_5}, \qquad S_2= S_{\alpha_1}\circ S_{\alpha_4}\circ S_{\alpha_6}.
\end{equation}
If we require an invariance of the Hamiltonian and the symplectic form \eqref{eq:*H-g} with respect to $S_1$,
and restrict on the set of admissible roots $\beta_k$, then we get a real Hamiltonian form  of the $E_6^{(1)}$ ATFT, described by:
\begin{equation}\label{eq:E6-rhf}\begin{split}
\mathcal{H}^{\bbbr}_{E_6} &= {1 \over 2} ({\bf p}^+ , {\bf p}^+ ) -{1 \over 2} ({\bf p}^- , {\bf p}^- ) +
\sum_{k =0,1,4,6} n_k' (e^{-({\bf q}^+ ,\alpha _k)} -1) + \sum_{k= 2,3,5} n_k'' ( \cos ({\bf q}^- ,\beta_k)-1),
\end{split}\end{equation}
where $n_1'=n_6'=1$, $n_4'=3$ and $n_2'' =n_3'' =n_5'' =2$ \cite{Bourb1}, and
\begin{equation}\label{eq:*ome-g6}\begin{split}
\omega ^{\bbbr}_{E_6^{(1)}} =
(\delta {\bf p}^+  \wedge \delta {\bf q}^+) -(\delta {\bf p}^-  \wedge \delta {\bf q}^-) .
\end{split}\end{equation}

\subsection{Affine Toda field theories related to $E_7^{(1)}$}\label{ssec:E7-1}

The extended root system of  this algebra is given by
\begin{equation}\label{eq:e7roots}\begin{aligned}
\alpha_1 &={1\over 2}(e_1-e_2-e_3-e_4-e_5-e_6-e_7+e_8), &\quad \alpha_2 &=e_1+e_2, \\
\alpha_3 &=e_2-e_1,\qquad \alpha_4=e_3-e_2, &\quad \alpha_5 &=e_4-e_3, \\
\alpha_6 &=e_5-e_4,\qquad \alpha_7=e_6-e_5, &\quad \alpha_0 &=e_7-e_8,
\end{aligned}\end{equation}
where $\alpha_1,\dots , \alpha_7$ form the set of simple roots of
$E_7$ and $\alpha_0$ is the minimal root of the algebra. The extended Dynkin diagram for $E_7^{(1)}$ is shown on Figure \ref{fig:Dynkin-E7}.
This is the standard def\/inition of the root system of $E_7
$ embedded into the 8-dimensional Euclidean space $\bbbe^8 $. The
root space $\bbbe_7 $ of the algebra $E_7$ is the
7-dimensional subspace of $\bbbe^8 $ orthogonal to the vector
$e_7+e_8 $. Thus any vector $\vec{q} $ belonging to $\bbbe_7 $ has
7 independent coordinates and can be written as:
\begin{gather}\label{eq:vec-q'}
\vec{q} = \sum_{k=1}^{6} q_k e_k + q_7 e'_7, \qquad e'_7 = {1\over
\sqrt{2} } (e_7 - e_8).
\end{gather}

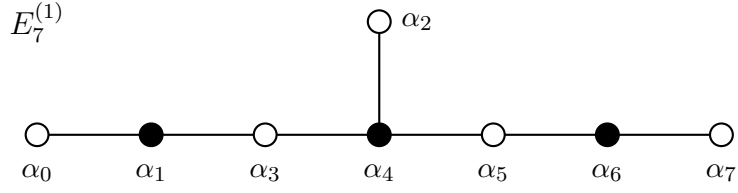
\begin{figure}
\begin{center}
\begin{tikzpicture}
\draw[thick] (-1.5,0) --(7.5,0);
\draw[thick] (3,0) --(3,1.5);
\draw[thick,fill=white] (-1.5,0) circle (1.5 mm);
\draw[thick,fill=black] (0,0) circle (1.5 mm);
\draw[thick,fill=white] (1.5,0) circle (1.5 mm);
\draw[thick,fill=black] (3,0) circle (1.5 mm);
\draw[thick,fill=white] (4.5,0) circle (1.5 mm);
\draw[thick,fill=black] (6,0) circle (1.5 mm);
\draw[thick,fill=white] (7.5,0) circle (1.5 mm);
\draw[thick,fill=white] (3,1.5) circle (1.5 mm);

\draw (-1.5,-0.5) node {\small $\alpha_0$};
\draw (0,-0.5) node {\small $\alpha_1$};
\draw (1.5,-0.5) node {\small $\alpha_3$};
\draw (3,-0.5) node {\small $\alpha_4$};
\draw (4.5,-0.5) node {\small $\alpha_5$};
\draw (6,-0.5) node {\small $\alpha_6$};
\draw (7.5,-0.5) node {\small $\alpha_7$};
\draw (3.5,1.5) node {\small $\alpha_2$};

\draw (-1.5,1.5) node {$E_7^{(1)}$};

\end{tikzpicture}
\end{center}
  \caption{\small The extended Dynkin diagram of the  complex untwisted affine Kac-Moody algebra $E_7^{(1)}$. The white roots are invariant with respect to the automorphism $S_1$ in \eqref{eq:DCox-E7}.}\label{fig:Dynkin-E7}
\end{figure}

\noindent
Using the dihedral realization of the Coxeter automorphism for $E_7^{(1)}$:
\begin{equation}\label{eq:DCox-E7}
{\rm Cox}\, (E_7^{(1)})=S_1\circ S_2, \qquad S_1= S_{\alpha_1}\circ S_{\alpha_4}\circ S_{\alpha_6}, \qquad S_2=  S_{\alpha_2}\circ S_{\alpha_3}\circ S_{\alpha_5}\circ S_{\alpha_7},
\end{equation}
we  impose an invariance of the Hamiltonian \eqref{eq:*H-g} with respect to the automorphism $S_1$ in \eqref{eq:DCox-E7}.
Then the root space ${\Bbb E}_7$ will split into direct sums: ${\Bbb E}_7={\Bbb E}_{7,+}\oplus {\Bbb E}_{7,-}$, with
\[
\alpha_2, \alpha_3, \alpha_5, \alpha_7 \in {\Bbb E}^{7,+}, \qquad \alpha_1, \alpha_4, \alpha_6 \in {\Bbb E}^{7,-}.
\]
If we again require an invariance of the Hamiltonian and the symplectic form \eqref{eq:*H-g} with respect to $S_1$ from \eqref{eq:DCox-E7}, and restrict on the set of admissible roots $\beta_k$, then we get a real Hamiltonian form  of the $E_7^{(1)}$ ATFT, described by:
\begin{equation}\label{eq:eq:E7-rhf}\begin{split}
\mathcal{H}^{\bbbr}_{E_7} &= {1 \over 2} ({\bf p}^+ , {\bf p}^+ ) -{1 \over 2} ({\bf p}^- , {\bf p}^- ) +
\sum_{k =0,2,3,5,7} n_k' (e^{-({\bf q}^+ ,\alpha _k)} -1) + \sum_{k= 1,4,6} n_k'' ( \cos ({\bf q}^- ,\beta_k)-1),
\end{split}\end{equation}
where $n_1'=n_6'=2$, $n_4'=4$  and $n_2'' =2$, $n_3''=n_5''=3$, $n_7'' =1$  \cite{Bourb1}, and
\begin{equation}\label{eq:*ome-g7}\begin{split}
\omega ^{\bbbr}_{E_7^{(1)}} =
(\delta {\bf p}^+  \wedge \delta {\bf q}^+) -(\delta {\bf p}^-  \wedge \delta {\bf q}^-) .
\end{split}\end{equation}

\subsection{Affine Toda field theories related to $E_8^{(1)}$}\label{ssec:E8-1}

The set of admissible roots for this algebra is
\begin{equation}\label{eq:e8roots}\begin{aligned}
\alpha_0 &=\frac{1}{2}(e_1 + e_2 + e_3 + e_4 + e_5 - e_6 - e_7 + e_8) \quad
\alpha_1 =\frac{1}{2}(e_1 - e_2 - e_3 - e_4 - e_5 - e_6 - e_7 + e_8) , \\
\alpha_2 &= e_1 + e_2, \quad
\alpha_3 = e_2 - e_1, \quad
\alpha_4 = e_3 - e_2, \\
\alpha_5 &= e_4 - e_3, \quad
\alpha_6 = e_5 - e_4, \quad
\alpha_7 = e_6 - e_5 \quad
\alpha_8 = e_7 - e_6.
\end{aligned}\end{equation}
where $\alpha_1,\dots , \alpha_8$ form the set of simple roots of
$E_8$ and $\alpha_0$ is the minimal root of the algebra.
If we exclude  $\alpha_8$  from \eqref{eq:e8roots}, we will get the set of positive roots \eqref{eq:e7roots} form $E_7$. The extended Dynkin diagram for $E_8^{(1)}$ is shown on Figure \ref{fig:Dynkin-E8}.

The fundamental weights of $E_8$ are
\begin{eqnarray*}
\omega_1&=& 2e_8, \qquad \omega_2 = {1\over 2}(e_1+e_2+e_3+e_4+e_5+e_6+e_7+5e_8),\\
\omega_3 &=& {1\over 2}(-e_1+e_2+e_3+e_4+e_5+e_6+e_7+7e_8), \qquad \omega_4= e_3+e_4+e_5+e_6+e_7+5e_8,\\
\omega_5 &=& e_4+e_5+e_6+e_7+4e_8, \qquad \omega_6 = e_5+e_6+e_7+3e_8,\\
\omega_7 &=& e_6+e_7+2e_8, \qquad \omega_8=e_7+e_8=-\alpha_0.
\end{eqnarray*}

\begin{figure}
\begin{center}
\begin{tikzpicture}
\draw[thick] (0,0) --(10.5,0);
\draw[thick] (3,0) --(3,1.5);

\draw[thick,fill=black] (0,0) circle (1.5 mm);
\draw[thick,fill=white] (1.5,0) circle (1.5 mm);
\draw[thick,fill=black] (3,0) circle (1.5 mm);
\draw[thick,fill=white] (4.5,0) circle (1.5 mm);
\draw[thick,fill=black] (6,0) circle (1.5 mm);
\draw[thick,fill=white] (7.5,0) circle (1.5 mm);
\draw[thick,fill=black] (9,0) circle (1.5 mm);
\draw[thick,fill=white] (10.5,0) circle (1.5 mm);
\draw[thick,fill=white] (3,1.5) circle (1.5 mm);

\draw (0,-0.5) node {\small $\alpha_1$};
\draw (1.5,-0.5) node {\small $\alpha_3$};
\draw (3,-0.5) node {\small $\alpha_4$};
\draw (4.5,-0.5) node {\small $\alpha_5$};
\draw (6,-0.5) node {\small $\alpha_6$};
\draw (7.5,-0.5) node {\small $\alpha_7$};
\draw (9,-0.5) node {\small $\alpha_8$};
\draw (10.5,-0.5) node {\small $\alpha_0$};
\draw (3.5,1.5) node {\small $\alpha_2$};

\draw (0,1.5) node {$E_8^{(1)}$};

\end{tikzpicture}
\end{center}
  \caption{\small The extended Dynkin diagram of the  complex untwisted affine Kac-Moody algebra $E_8^{(1)}$. The white roots are invariant with respect to the automorphism $S_2$ in \eqref{eq:DCox-E8}.}\label{fig:Dynkin-E8}
\end{figure}
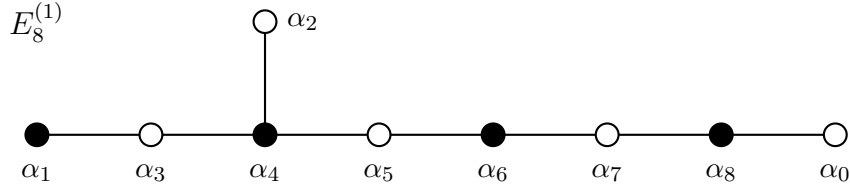

\noindent
We will use the dihedral realization of the Coxeter automorphism for $E_8^{(1)}$:
\begin{equation}\label{eq:DCox-E8}
{\rm Cox}\, (E_8^{(1)})=S_1\circ S_2, \qquad S_1= S_{\alpha_1}\circ S_{\alpha_4}\circ S_{\alpha_6}\circ S_{\alpha_8}, \qquad S_2= S_{\alpha_2}\circ S_{\alpha_3}\circ S_{\alpha_5}\circ S_{\alpha_7}.
\end{equation}
We  impose an invariance of the Hamiltonian \eqref{eq:*H-g} with respect to the automorphism $S_1$ in \eqref{eq:DCox-E7}.
Then the root space ${\Bbb E}_8$ will split into direct sums: ${\Bbb E}_8={\Bbb E}_{8,+}\oplus {\Bbb E}_{8,-}$, with
\[
\alpha_2, \alpha_3, \alpha_5, \alpha_7 \in {\Bbb E}_{8,+}, \qquad \alpha_1, \alpha_4, \alpha_6 ,\alpha_8\in {\Bbb E}_{8,-}.
\]
The invariance of the Hamiltonian and the symplectic form \eqref{eq:*H-g} with respect to $S_1$ from \eqref{eq:DCox-E8}, and restricting on the set of admissible roots $\beta_k$, will result in a real Hamiltonian form  of the $E_8^{(1)}$ ATFT, described by:
\begin{equation}\label{eq:eq:E8-rhf}\begin{split}
\mathcal{H}^{\bbbr}_{E_8} &= {1 \over 2} ({\bf p}^+ , {\bf p}^+ ) -{1 \over 2} ({\bf p}^- , {\bf p}^- ) +
\sum_{k =0,2,3,5,7} n_k' (e^{-({\bf q}^+ ,\alpha _k)} -1) + \sum_{k= 1,4,6,8} n_k'' ( \cos ({\bf q}^- ,\beta_k)-1),
\end{split}\end{equation}
where $n_1''=n_8''=2$, $n_4''=6, n_6''=4$ and $n_2'=n_7'=3$, $n_3' =4$, $n_5' =5$ \cite{Bourb1}, and
\begin{equation}\label{eq:*ome-g8}\begin{split}
\omega ^{\bbbr}_{E_8^{(1)}} =
(\delta {\bf p}^+  \wedge \delta {\bf q}^+) -(\delta {\bf p}^-  \wedge \delta {\bf q}^-) .
\end{split}\end{equation}

\subsection{Affine Toda field theories related to $F_4^{(1)}$}\label{ssec:F4-1}

The extended roots of this algebra are
\begin{equation}\label{eq:f4roots}\begin{aligned}
\alpha_1&= e_2-e_3, \qquad \alpha_2 = e_3-e_4,  &\quad \alpha_3&= e_4, \\
 \alpha_4&= {1\over 2}(e_1-e_2-e_3-e_4), &\quad \alpha_0&=-e_1-e_2,
\end{aligned}\end{equation}
with $\alpha_0$ being the minimal root. The extended Dynkin diagram of $F_4^{(1)}$ is shown on Figure \ref{fig:Dynkin-F4}.
The fundamental weights of $F_4$ are
\begin{equation}\label{eq:weif4}\begin{aligned}
  \omega_1 &= e_1+e_2, &\quad \omega_2 &= 2e_1+e_2+e_3, \\
  \omega_3 &= {1\over 2}(3e_1+e_2+e_3+e_4), &\quad \omega_4 &=e_1.
\end{aligned}\end{equation}

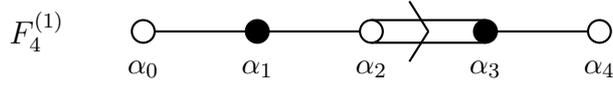
\begin{figure}
\begin{center}
\begin{tikzpicture}
\draw[thick] (0,0) --(3,0);
\draw[thick] (4.5,0) --(6,0);
\draw[thick] (3,0.15) --(4.5,0.15);
\draw[thick] (3,-0.15) --(4.5,-0.15);

\draw[thick] (3.5,-0.4) --(3.75,0);
\draw[thick] (3.5,0.4) --(3.75,0);

\draw[thick,fill=white] (0,0) circle (1.5 mm);
\draw[thick,fill=black] (1.5,0) circle (1.5 mm);
\draw[thick,fill=white] (3,0) circle (1.5 mm);
\draw[thick,fill=black] (4.5,0) circle (1.5 mm);
\draw[thick,fill=white] (6,0) circle (1.5 mm);

\draw (0,-0.5) node {\small $\alpha_0$};
\draw (1.5,-0.5) node {\small $\alpha_1$};
\draw (3,-0.5) node {\small $\alpha_2$};
\draw (4.5,-0.5) node {\small $\alpha_3$};
\draw (6,-0.5) node {\small $\alpha_4$};

\draw (-1.4,0) node {$F_4^{(1)}$};

\end{tikzpicture}
\end{center}
  \caption{\small The extended Dynkin diagram of the  complex untwisted affine Kac-Moody algebra $F_4^{(1)}$. The white roots are invariant with respect to the automorphism $S_2$ in \eqref{eq:DCox-F4}.}\label{fig:Dynkin-F4}
\end{figure}

\noindent
 Using the dihedral realisation of the Coxeter automorphism for $F_4^{(1)}$:
\begin{equation}\label{eq:DCox-F4}
{\rm Cox}\, (F_4^{(1)})=S_1\circ S_2, \qquad S_1=  S_{\alpha_1}\circ S_{\alpha_3}, \qquad S_2= S_{\alpha_2}\circ S_{\alpha_4},
\end{equation}
and imposing  an invariance of the Hamiltonian \eqref{eq:*H-g} with respect to the automorphism $S_1$ in \eqref{eq:DCox-F4}, will extract a real Hamiltonian form of $F_4^{(1)}$ ATFT of the form
\begin{equation}\label{eq:F4-RHF}\begin{split}
\mathcal{H}^{\bbbr}_{F_4^{(1)}} =
{1 \over 2} \sum_{k=1}^{4} \left((p_k^+)^2- (p_k^-)^2\right) +
\sum_{k =0,2,4} n_k' (e^{-({\bf q}^+ ,\alpha _k)} -1) + \sum_{k= 1,3} n_k'' ( \cos ({\bf q}^- ,\beta_k)-1),
\end{split}\end{equation}
where $n_2'=3$,  $n_4'=2$ and  $n_1''=2$,  $n_3''=4$. In addition
\begin{equation}\label{eq:*ome-g4}\begin{split}
\omega ^{\bbbr}_{F_4^{(1)}} =
(\delta {\bf p}^+  \wedge \delta {\bf q}^+) -(\delta {\bf p}^-  \wedge \delta {\bf q}^-) .
\end{split}\end{equation}
Then the root space ${\Bbb E}_4$ will split into direct sums: ${\Bbb E}_4={\Bbb E}_{4,+}\oplus {\Bbb E}_{4,-}$, with
$\alpha_2, \alpha_4  \in {\Bbb E}_{4,+}$ and $\alpha_1, \alpha_3 \in {\Bbb E}_{4,-}$.

\subsection{Affine Toda field theories related to $G_2^{(1)}$}\label{ssec:G2-1}

The set of admissible roots for this algebra is given by
\begin{eqnarray}\label{eq:g2roots}
\alpha_1&=& e_1-e_2, \qquad \alpha_2= -e_1+e_2+e_3, \qquad \alpha_0=-e_1-e_2+2e_3.
\end{eqnarray}
The fundamental weights of $G_2$ are
\begin{eqnarray*}
  \omega_1 &=& e_1-e_2+2e_3, \qquad \omega_2= -e_1-e_2+2e_3.
\end{eqnarray*}
The extended Dynkin diagram of $G_2^{(1)}$ is shown on Figure \ref{fig:Dynkin-G2}.

Let us take $C_1=S_{\alpha_2}$.
The invariance of the Hamiltonian \eqref{eq:*H-g} with respect to the automorphism $C_1$ in will give a real Hamiltonian form of $G_2^{(1)}$ ATFT described by
\begin{eqnarray}\label{eq:G2-RHF}
\mathcal{H}^{\bbbr}_{G_2^{(1)}} &=&
{1 \over 2} \sum_{k=1}^{r} \left((p_k^+)^2- (p_k^-)^2\right) + (e^{-({\bf q}^+ ,\alpha_0)}-1) +
 3(e^{-({\bf q}^+ ,\alpha_1)}-1)  + 2 (\cos ({\bf q}^-,\alpha_2)-1),\\
\label{eq:*ome-g2}
\omega ^{\bbbr}_{G_2^{(1)}} &=&
(\delta {\bf p}^+  \wedge \delta {\bf q}^+) -(\delta {\bf p}^-  \wedge \delta {\bf q}^-) ,
\end{eqnarray}
The root space ${\Bbb E}_2$ will split into direct sums: ${\Bbb E}_2={\Bbb E}_{2,+}\oplus {\Bbb E}_{2,-}$, with
$\alpha_0, \alpha_1  \in {\Bbb E}_{2,+}$ and $\alpha_2 \in {\Bbb E}_{2,-}$.

\begin{figure}
\begin{center}
\begin{tikzpicture}

\draw[thick] (0,0.15) --(1.5,0.15);
\draw[thick] (0,0) --(1.5,0);
\draw[thick] (0,-0.15) --(1.5,-0.15);
\draw[thick] (1.5,0) --(3,0);

\draw[thick] (0.85,0.4) --(0.6,0);
\draw[thick] (0.85,-0.4) --(0.6,0);

\draw[thick,fill=white] (0,0) circle (1.5 mm);
\draw[thick,fill=black] (1.5,0) circle (1.5 mm);
\draw[thick,fill=white] (3,0) circle (1.5 mm);

\draw (0,-0.5) node {\small $\alpha_1$};
\draw (1.5,-0.5) node {\small $\alpha_2$};
\draw (3,-0.5) node {\small $\alpha_0$};

\draw (-1.4,0) node {$G_2^{(1)}$};

\end{tikzpicture}
\end{center}
  \caption{\small The extended Dynkin diagram of the  complex untwisted affine Kac-Moody algebra $G_2^{(1)}$.
  }\label{fig:Dynkin-G2}
\end{figure}
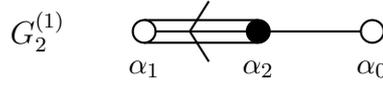


\section{On the spectral properties of the Lax operators with $\mathbb{Z}_h $-reduction}

\subsection{General theory}

The  operator $L$ (\ref{eq:2.1}) is very convenient for deriving the
2-dimensional TFT. However this formulation is not convenient  for
constructing the fundamental analytic solutions (FAS) of $L$.  The first
step we will do will be to apply on $L$ a similarity transformation
\begin{equation}\label{eq:Ltil}\begin{split}
 \tilde{L}\tilde{\psi} = g_0^{-1} L g_0 \tilde{\psi} =  i \frac{\partial \tilde{\psi}}{ \partial x } -
\left( i \sum_{k=0}^{6} q_k \mathcal{E}_k + \lambda \tilde{J}_0 \right) \tilde{\psi} =0
\end{split}\end{equation}
where $\tilde{J}_0 \in \mathfrak{h}$; i.e.
\[ \tilde{J}_0 = \sum_{ k=1}^{6}y_j H_{\alpha_j} =\sum_{ k=1}^{6}z_j H_{\omega_j}  \]
where $y_j$ and $z_j$ must be expressed in terms of $\omega = \exp( 2\pi i/h_g)$.

Let us somewhat simplify the notations in the two Lax operators

\begin{equation}\label{eq:LtL}\begin{split}
L\psi &\equiv \frac{\partial \psi}{ \partial x } + (\vec{q}_x - \lambda J_0 )\psi(x,t,\lambda) =0, \\
\tilde{L}\tilde{\psi}  &\equiv \frac{\partial \tilde{\psi} }{ \partial x } + (Q(x,t) - \lambda \tilde{J} )\tilde{\psi} (x,t,\lambda) =0,
\end{split}\end{equation}
where $J_0 = \sum_{\alpha \in \delta_0}^{} E_\alpha$, $\delta_0$ is the set of simple roots $\alpha_j$ and the minimal root
$\alpha_0$, and $ \tilde{J} \in \mathfrak{h}$ and $\omega = \exp(2\pi i/h)$.

The spectral theory of the Lax operators with $\mathbb{Z}_h$ Mikhailov reduction groups \cite{Mikh}
for the classical series of simple Lie algebras are by now well developed, see \cite{GeYa,SIAM*14,VG-Ya-13}
Here we will formulate them taking into account when necessary the peculiarities of the exceptional algebras.

In \cite{TJ22} the problem for the interrelation between $L$ and $\tilde{L}$ was solved by using the Chevallie
basis for $A_5^{(1)}$ \cite{DriSok}. Indeed, for the classical series Chevallie basis for the typical representations are well
known so it is not difficult to diagonalize $J_0$ and then evaluate $\alpha_j(\tilde{J}_0)$. The same procedure for the exceptional algebras
(except $G_2$) is rather more complicated. Indeed, the bases for the typical representations of $F_4$, $E_6$ and $E_7$ are well
known, see \cite{Howlett} where one can find all root-vectors of $F_4$ as $26\times 26$ matrices. For $E_6$
and $E_7$ the simple roots have been given in \cite{Howlett}; it is possible that all root vectors for all exceptional algebras
have been evaluated in the framework of MAGMA project. So in general one can find $J_0$ for the typical representations
of all exceptional algebras. Then the construction of $\tilde{J}_0$ becomes the next task. The eigenvalues of $J_0$
and the eigenvectors may be calculated, though the latter would be rather involved expressions.
The next challenge would be to constructing out of the properly normalized eigenvectors $g_0$ and to ensure that it
is an element of the corresponding exceptional group.

The easiest way to get a realization of the Coxeter automorphism $C_g$ is to represent it as an element
of the Cartan subgroup:
\begin{equation}\label{eq:Cox0}\begin{split}
C_g = \exp \left( \frac{2\pi i}{h_g} \sum_{j=1}^{r} H_{\omega_j} \right),
\end{split}\end{equation}
where $h_g$ is the Coxeter number and $\omega_j$ are the fundamental weights of the algebra $\mathfrak{g}$. Indeed,
using the Cartan-Weyl commutation relations \cite{Helg} one can easily check that
\begin{equation}\label{eq:Cox3}\begin{split}
 C_g J_0 C_g^{-1} = \omega J_0, \qquad \omega = \exp \left( \frac{2\pi i}{h_g} \right).
\end{split}\end{equation}
This construction is compatible with the Lax operator $L$, since $C_g $ commutes with $\vec{q}_x$ and therefore
\begin{equation}\label{eq:CoxL}\begin{split}
 C_g L (\lambda) C_g^{-1} =  L (\omega\lambda).
\end{split}\end{equation}
Treating the operator $\tilde{L}$ requires the use of the Coxeter automorphism $\tilde{C}_g$ realized as
element of the Weyl group  $\mathcal{W}_g$. We already mentioned the difficulties involved with the standard approach based on the
use of the Chevallie basis. Therefore, treating the exceptional algebras  we will use a different approach.
Our idea is to work only in the Euclidean space $\mathbb{E}_r$ in which the root systems are embedded.
Indeed, there we can construct $\tilde{C}_g$ using its dihedral form \cite{Helg, Bourb1, Cart}:
\begin{equation}\label{eq:Cox1}\begin{split}
\tilde{C}_g = S_w \circ S_b, \qquad S_w = \prod_{\alpha \in W_g}^{} S_\alpha, \qquad S_b = \prod_{\beta \in B_g}^{} S_\beta,
\end{split}\end{equation}
where $W_g$ (respectively $B_g$) are the sets of `white` roots (resp. `black` roots) of the Dinkin diagram of the
algebra $\mathfrak{g}$. Note that all elements of $W_g$ (resp. $B_g$) are all orthogonal, so $S_w$ and $S_b$ are
Weyl reflections, i.e.  $S_w^2 =\openone$ and $S_b^2 =\openone$. Then $C_g^{h_g} =\openone$, where $h_g$ is the Coxeter
number of the algebra $\mathfrak{g}$. Thus we will find the realization of the Coxeter automorphism as element of the
Weyl group $\mathcal{W}_g$ of the algebra.

The next step should be to construct the Cartan subalgebra element $\tilde{J}_0$. In fact we will replace
this step, and instead of constructing $\tilde{J}_0$ we will construct the vector $\vec{v}_g \in \mathbb{E}_r$ which is
dual to $\tilde{J}_0$. Then we will us the fact that $\alpha_j(\tilde{J}_0) = (\alpha_j, \vec{v}_g)$. This is all we need
in order to construct the continuous spectrum and the spectral data for $\tilde{L}$. On the other hand, since  $\tilde{L}$
is related to $L$ by the similarity transformation  $\tilde{L} =g_0 L g_0^{-1}$ then both operators will have the
same spectral data.

The construction of the vector $\vec{v}_g$ is based on its basic property of invariance under the $\mathbb{Z}_h$
Mikhailov reduction group; like in (\ref{eq:CoxL}) we must have:
\begin{equation}\label{eq:CoxtL}\begin{split}
\tilde{C}_g \tilde{L} (\lambda) \tilde{C}_g^{-1} =  \tilde{L} (\omega\lambda), \qquad \mbox{i.e.} \qquad
  \tilde{C}_g(\vec{v}_g) = \omega \vec{v}_g.
\end{split}\end{equation}
Such invariant vectors can be obtained by taking weighted averaged action by $\tilde{C}_g$:
\begin{equation}\label{eq:vg}\begin{split}
 \vec{v}_{0} = \sum_{s=0}^{h_g-1} \omega^{-s} \tilde{C}_g^s (\vec{v}_g).
\end{split}\end{equation}
Obviously there is an arbitrariness in the choice of $\vec{v}_g$: it can be a root, or fundamental weight  of the algebra
$\mathfrak{g}$. At the same time we can use the fact that the spectral parameter $\lambda$ can be redefined, so
this arbitrariness can be taken care of.

\subsection{The FAS of the Lax operators $L$}\label{sec:2.3}

Here we just outline the procedure of constructing the fundamental analytic solutions
(FAS) of $L$ \cite{BeCo1, BeCo2,  SIAM*14, GeYa}. First we have to determine the regions of analyticity. For
smooth potentials $Q(x) $ that fall off fast enough for $x\to \pm \infty$ these regions
are the $2h $ sectors $ \Omega _\nu  $ separated by the rays
$l_\nu  $ on which $\re \lambda (\alpha, \vec{v}_g) =0 $, where by $\alpha$ is a root of $\mathfrak{g}$.  The rays $l_\nu $ and the sectors $\Omega_\nu$ are given by:
\begin{eqnarray}\label{eq:l-nu}
l_\nu \colon \arg(\lambda) ={\pi (\nu -1) \over h_g } , \qquad \Omega_\nu \colon \frac{\pi (\nu -1)}{h_g} \leq \arg(\lambda) \leq \frac{\nu}{h_g},
\qquad \nu =1,\dots , 2h,
\end{eqnarray}
and close angles equal to $\pi/h $.

The main result is that in each of the sectors $\Omega_\nu$ we can construct fundamental analytic solutions
(FAS) $\xi_\nu (x,t,\lambda)$  of the Lax operator $\mathcal{L}$:
\begin{equation}\label{eq:calL}\begin{split}
 \mathcal{L}\xi_\nu (x,t,\lambda) \equiv i \frac{\partial \xi_\nu }{ \partial x } +Q_x \xi_\nu (x,\lambda) - \lambda [\tilde{J},\xi_\nu]=0, \\
 Q_x = \sum_{j=1}^{r} q_{j,x} \mathcal{O}_j, \qquad \mathcal{O}_j = \sum_{s=0}^{h_g-1} C_g^s(E_{\alpha_j}).
\end{split}\end{equation}

 The asymptotics of $\xi^\nu (x,\lambda ) $ and $\xi^{\nu -1} (x,\lambda )  $ along the
ray $l_\nu  $ can be written in the form \cite{67,GeYa}:
\begin{equation}\label{eq:xi-as}\begin{aligned}
\lim_{x\to -\infty } e^{-\lambda \tilde{J} x} \xi^\nu (x,\lambda e^{i0} )
e^{\lambda \tilde{J} x} &= S_\nu ^+ (\lambda ), &\quad \lambda \in l_\nu ,\\
\lim_{x\to -\infty } e^{-\lambda \tilde{J} x} \xi^{\nu -1} (x,\lambda e^{-i0} )
e^{\lambda \tilde{J} x} &= S_{\nu} ^- (\lambda ), &\quad \lambda \in l_{\nu },\\
\lim_{x\to \infty } e^{-\lambda \tilde{J} x} \xi^\nu (x,\lambda e^{i0} )
e^{\lambda \tilde{J} x} &= T_\nu ^-D_\nu ^+ (\lambda ), &\quad \lambda \in l_\nu ,\\
\lim_{x\to \infty } e^{-\lambda \tilde{J} x} \xi^{\nu -1} (x,\lambda e^{-i0} )
e^{\lambda \tilde{J} x} &= T_{\nu} ^+D_{\nu } ^-(\lambda ), &\quad \lambda \in l_{\nu } ,
\end{aligned}\end{equation}
where the matrices $S_\nu ^+ $, $T_\nu ^+ $ (resp. $S_\nu ^- $, $T_\nu ^-
$) are upper-triangular (resp. lower-triangular) with respect to the $\nu
$-ordering. They provide the Gauss decomposition
of the scattering matrix with respect to the $\nu  $-ordering, i.e.:
\begin{equation}\begin{split}\label{eq:nu-gauss}
T_\nu (\lambda ) &= T_\nu ^-(\lambda ) D_\nu ^+(\lambda ) \hat{S}_\nu ^+(\lambda ), \qquad \lambda \in l_\nu e^{i0}, \\
&=  T_\nu ^+(\lambda ) D_\nu ^-(\lambda ) \hat{S}_\nu ^- (\lambda ) ,  \qquad \lambda \in l_\nu e^{-i0}.
\end{split}\end{equation}
More careful analysis shows \cite{GeYa} that in fact $T_\nu
(\lambda ) $ belongs to a subgroup ${\cal G}_\nu$ of $SL(N,\mathbb{C})$.
In particular we can choose $\lambda$ in such a way that the simple roots of $\mathfrak{g}$
will be related to the rays $l_0$ and $l_1$; more precisely the `black` roots will be related to
$l_0$ while the `white` roots will be related to $l_1$. In particular this means that with
$l_0$ and $l_1$ we can relate subalgebras of $\mathfrak{g}$ which will be direct sums of several
copies of $sl(2)$. In addition we can introduce the analogs of the reflection coefficients $\rho^\pm_{\nu,j}$
and $\tau^\pm_{\nu,j}$ with $\nu=0,1$ as follows:
\begin{equation}\label{eq:58}\begin{aligned}
 S_0^\pm(\lambda) &=\exp \left( \sum_{\alpha \in B_g}^{}\tau_{0,\alpha}^\pm (\lambda) E_{\pm \alpha}\right), &\quad
 T_0^\mp(\lambda) &=\exp \left( \sum_{\alpha \in B_g}^{}\rho_{0,\alpha}^\pm (\lambda) E_{\mp \alpha}\right),  \\
 S_1^\pm (\lambda) &=\exp \left( \sum_{\alpha \in W_g}^{}\tau_{0,\alpha}^\pm (\lambda) E_{\pm \alpha}\right), &\quad
 T_1^\mp (\lambda)&=\exp \left( \sum_{\alpha \in W_g}^{}\rho_{0,\alpha}^\pm (\lambda) E_{\mp \alpha}\right),&\quad \\
 D_0^\pm(\lambda) &= \exp \left( \sum_{j=1}^{r} \frac{2d_{0,j}^\pm (\lambda) }{(\alpha_j,\alpha_j)}H_{\alpha_j}\right), &\quad
 D_1^\pm(\lambda) &= \exp \left( \sum_{j=1}^{r} \frac{2d_{0,j}^\pm (\lambda)}{(\alpha_j,\alpha_j)}H_{\alpha_j}\right).
\end{aligned}\end{equation}
Next the $\mathbb{Z}_h $-symmetry allows us to extend these scattering data to the other rays as well. Indeed applying the Coxeter
automorphism we get first the relations between FAS:
\begin{equation}\label{eq:Z_n-cons}\begin{aligned}
 C_g(\xi^\nu (x,\lambda)) &= \xi^{\nu +2} (x,\lambda\omega ) ,
&\quad C_g( T_\nu^\pm (\lambda  ))  &= T_{\nu +2}(\lambda\omega ), \\
C_g (S^\pm_\nu (\lambda  ))  &= S^\pm_{\nu +2}(\lambda \omega),
&\quad C_g( D^\pm_\nu (\lambda)) &= D^\pm_{\nu +2}(\lambda \omega) ,
\end{aligned}\end{equation}
where the index $\nu +2 $ should be taken modulo $2h $.
Consequently we can view as independent only the data on two of the
rays, e.g. on $l_1 $ and $ l_0 $; all the rest will be
recovered using (\ref{eq:Z_n-cons}).

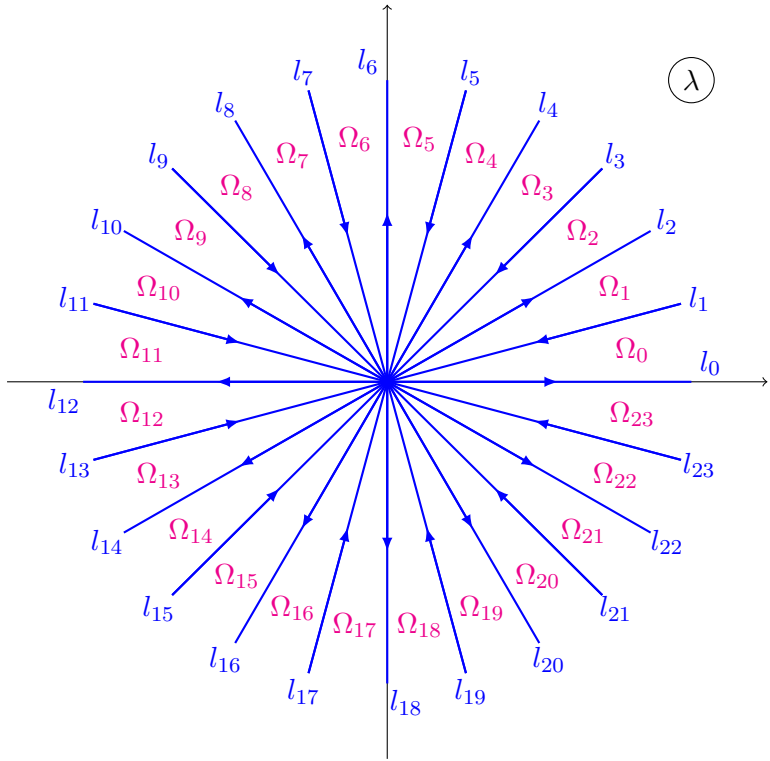
\begin{figure}
\begin{center}
\begin{tikzpicture}
\draw[->] (-5,0) -- (5,0);
\draw[->] (0,-5) -- (0,5);

\draw (4,4) circle (3mm);
\draw (4,4) node {\small $\lambda$};

\draw [blue,thick] (-4,0) -- (4,0);
\draw [blue,thick,latex-latex] (-2.25,0) -- (2.25,0);

\draw [blue,thick,rotate=15] (-4,0) -- (4,0);
\draw [blue,thick,-latex,rotate=15] (-4,0) -- (-2,0);
\draw [blue,thick,latex-,rotate=15] (2,0) -- (4,0);

\draw [blue,thick,rotate=30] (-4,0) -- (4,0);
\draw [blue,thick,latex-latex,rotate=30] (-2.25,0) -- (2.25,0);

\draw [blue,thick,rotate=45] (-4,0) -- (4,0);
\draw [blue,thick,-latex,rotate=45] (-4,0) -- (-2,0);
\draw [blue,thick,latex-,rotate=45] (2,0) -- (4,0);

\draw [blue,thick,rotate=60] (-4,0) -- (4,0);
\draw [blue,thick,latex-latex,rotate=60] (-2.25,0) -- (2.25,0);

\draw [blue,thick,rotate=75] (-4,0) -- (4,0);
\draw [blue,thick,-latex,rotate=75] (-4,0) -- (-2,0);
\draw [blue,thick,latex-,rotate=75] (2,0) -- (4,0);

\draw [blue,thick,rotate=90] (-4,0) -- (4,0);
\draw [blue,thick,latex-latex,rotate=90] (-2.25,0) -- (2.25,0);

\draw [blue,thick,rotate=105] (-4,0) -- (4,0);
\draw [blue,thick,-latex,rotate=105] (-4,0) -- (-2,0);
\draw [blue,thick,latex-,rotate=105] (2,0) -- (4,0);

\draw [blue,thick,rotate=120] (-4,0) -- (4,0);
\draw [blue,thick,latex-latex,rotate=120] (-2.25,0) -- (2.25,0);

\draw [blue,thick,rotate=135] (-4,0) -- (4,0);
\draw [blue,thick,-latex,rotate=135] (-4,0) -- (-2,0);
\draw [blue,thick,latex-,rotate=135] (2,0) -- (4,0);

\draw [blue,thick,rotate=150] (-4,0) -- (4,0);
\draw [blue,thick,latex-latex,rotate=150] (-2.25,0) -- (2.25,0);

\draw [blue,thick,rotate=165] (-4,0) -- (4,0);
\draw [blue,thick,-latex,rotate=165] (-4,0) -- (-2,0);
\draw [blue,thick,latex-,rotate=165] (2,0) -- (4,0);


\draw[blue] (4.25,0.25) node {\small $l_0$};
\draw[blue] (-0.25,4.25) node {\small $l_6$};
\draw[blue] (-4.25,-0.25) node {\small $l_{12}$};
\draw[blue] (0.25,-4.25) node {\small $l_{18}$};

\draw[blue] (15:4.25) node {\small $l_{1}$};
\draw[blue] (30:4.25) node {\small $l_{2}$};
\draw[blue] (45:4.25) node {\small $l_{3}$};
\draw[blue] (60:4.25) node {\small $l_{4}$};
\draw[blue] (75:4.25) node {\small $l_{5}$};

\draw[blue] (105:4.25) node {\small $l_{7}$};
\draw[blue] (120:4.25) node {\small $l_{8}$};
\draw[blue] (135:4.25) node {\small $l_{9}$};
\draw[blue] (150:4.25) node {\small $l_{10}$};
\draw[blue] (165:4.25) node {\small $l_{11}$};

\draw[blue] (195:4.25) node {\small $l_{13}$};
\draw[blue] (210:4.25) node {\small $l_{14}$};
\draw[blue] (225:4.25) node {\small $l_{15}$};
\draw[blue] (240:4.25) node {\small $l_{16}$};
\draw[blue] (255:4.25) node {\small $l_{17}$};

\draw[blue] (285:4.25) node {\small $l_{19}$};
\draw[blue] (300:4.25) node {\small $l_{20}$};
\draw[blue] (315:4.25) node {\small $l_{21}$};
\draw[blue] (330:4.25) node {\small $l_{22}$};
\draw[blue] (345:4.25) node {\small $l_{23}$};

\draw[magenta] (7.5:3.25) node {\small $\Omega_{0}$};
\draw[magenta] (22.5:3.25) node {\small $\Omega_{1}$};
\draw[magenta] (37.5:3.25) node {\small $\Omega_{2}$};
\draw[magenta] (52.5:3.25) node {\small $\Omega_{3}$};
\draw[magenta] (67.5:3.25) node {\small $\Omega_{4}$};
\draw[magenta] (82.5:3.25) node {\small $\Omega_{5}$};

\draw[magenta] (97.5:3.25) node {\small $\Omega_{6}$};
\draw[magenta] (112.5:3.25) node {\small $\Omega_{7}$};
\draw[magenta] (127.5:3.25) node {\small $\Omega_{8}$};
\draw[magenta] (142.5:3.25) node {\small $\Omega_{9}$};
\draw[magenta] (157.5:3.25) node {\small $\Omega_{10}$};
\draw[magenta] (172.5:3.25) node {\small $\Omega_{11}$};

\draw[magenta] (187.5:3.25) node {\small $\Omega_{12}$};
\draw[magenta] (202.5:3.25) node {\small $\Omega_{13}$};
\draw[magenta] (217.5:3.25) node {\small $\Omega_{14}$};
\draw[magenta] (232.5:3.25) node {\small $\Omega_{15}$};
\draw[magenta] (247.5:3.25) node {\small $\Omega_{16}$};
\draw[magenta] (262.5:3.25) node {\small $\Omega_{17}$};

\draw[magenta] (277.5:3.25) node {\small $\Omega_{18}$};
\draw[magenta] (292.5:3.25) node {\small $\Omega_{19}$};
\draw[magenta] (307.5:3.25) node {\small $\Omega_{20}$};
\draw[magenta] (322.5:3.25) node {\small $\Omega_{21}$};
\draw[magenta] (337.5:3.25) node {\small $\Omega_{22}$};
\draw[magenta] (352.5:3.25) node {\small $\Omega_{23}$};

\end{tikzpicture}

\end{center}
  \caption{\small The continuous spectrum of the Lax operators related to the algebras $E_6$ and $F_4$
  whose Coxeter nubers equal 12. Obviously the continuous spectra of for $E_8$, $E_7$ and $G_2$
  will have different number of rays equal to 60, 36 and 12 respectively.}\label{fig:Spectrum}
\end{figure}

\subsection{The inverse scattering problem and the Riemann-Hilbert problem}\label{sec:ism-rhp}

The next important step is the possibility to reduce the solution of the
inverse scattering problem (ISP) for the generalized Zakharov-Shabat system  to a (local) RHP. More precisely, we have:
\begin{equation}\label{eq:*-nu}\begin{aligned}
\xi^\nu (x,t,\lambda ) &= \xi^{\nu -1}(x,t,\lambda ) G_\nu (x,t,\lambda ), &\; \lambda &\in l_\nu ,\\
G_\nu (x,t,\lambda ) &= e^{\lambda \tilde{J} x -\lambda^{-1} \tilde{I} t} G_{0,\nu }(\lambda ) e^{-\lambda \tilde{J} x +\lambda^{-1} \tilde{I} t},  &\; G_{0,\nu
}(\lambda ) &= \left. \hat{S}_\nu ^- S_\nu ^+(\lambda )\right|_{t=0} .
\end{aligned}\end{equation}
where $ \tilde{I} = g_0^{-1} I(x,t) g_0$
The collection of all relations (\ref{eq:*-nu}) for $\nu =1,2,\dots,2h $
together with
\begin{equation}\label{eq:*-nu-norm}
\lim_{\lambda \to\infty } \xi^\nu (x,t,\lambda ) = \openone ,
\end{equation}
can be viewed as a local RHP posed on the collection of rays
$\Sigma \equiv \{l_\nu \}_{\nu =1}^{2h} $ with canonical normalization.
Rather straightforwardly we can  prove that if $\xi^\nu (x,\lambda ) $ is a solution of the
RHP (\ref{eq:*-nu}), (\ref{eq:*-nu-norm}) then $\chi ^\nu (x,\lambda
)=\xi^\nu (x,\lambda ) e^{-\lambda U_1 x} $ is a FAS of $L$ with potential
\begin{equation}\label{eq:q-CBC}
Q_x(x,t) = \lim_{\lambda \to\infty } \lambda \left( \tilde{J} - \xi^\nu (x,t,\lambda ) \tilde{J} \hat{\xi}^\nu (x,t,\lambda ) \right).
\end{equation}
For the classical Lie algebras the analyticity properties of  $d_j^\pm(\lambda ) $ allow one to
reconstruct  them from the sewing function $G(\lambda ) $ (\ref{eq:*-nu}) and from the and from the locations of their  zeroes and poles
\cite{GeYa, GVYa}. The idea was to consider the relations (\ref{eq:nu-gauss}) in the $j$-th fundamental representation of $\mathfrak{g}$
and take the matrix elements between the highest $|\omega_j^+ \rangle$ (resp. lowest  $|\omega_j^- \rangle$) weight vectors. This can be
done also for exceptional groups with the results
\begin{equation}\label{eq:dpm}\begin{aligned}
 d_{\nu,j}^+(\lambda) -  d_{\nu,j}^-(\lambda) & = -\ln \langle \omega_j^+ | \hat{T}_\nu^+ T_\nu^-(\lambda) | \omega_j^+\rangle \\
 &= -\ln \langle \omega_j^- | \hat{S}_\nu^- S_\nu^+(\lambda) | \omega_j^+\rangle .
\end{aligned}\end{equation}
If $\mathfrak{g}$ is one of the classical Lie algebras, then the right hand sides of (\ref{eq:dpm}) can be related to the
principal upper and lower minors of the scattering matric $T(\lambda)$ in the
typical representation.  For the exceptional algebras the fundamental representations can not so easy be related to the typical
(or lowest dimensional) one. So recovering the generating functionals of integrals of motion from their analyticity properties is an
open problem.

\section{The minimal sets of scattering data}

\subsection{The minimal sets of scattering data from the RHP}

As a consequence of the results in Section \ref{sec:ism-rhp}  we conclude that the following theorem holds true:
\begin{theorem}\label{thm:1}
Let us assume that the  Lax operator $\mathcal{L}$ (\ref{eq:calL}) is constructed along the
ideas in Section 4 above, i.e.  it is related to a simple Lie algebra $\mathfrak{g}$ and possesses
$\mathbb{Z}_h$ symmetry where $h$ is the Coxeter number of $\mathfrak{g}$.
Let us also assume that its potential $Q_x$ is such that $\mathcal{L}$ has
no discrete eigenvalues. Then its FAS $\xi^\nu (x,t,\lambda)$ in the sector $\Omega_\nu$
will provide a regular solution of the Riemann-Hilbert Problem  (\ref{eq:*-nu}).
Let us denote by $\delta_0$ and  $\delta_1$ the subsets of roots of the algebra $\mathfrak{g}$ which satisfy the equations:
\begin{equation}\label{eq:del10}\begin{aligned}
\alpha \in \delta_0 \qquad \mbox{if} \qquad \im \;\lambda (\alpha, \vec{v}_g) =0  \qquad \mbox{for} \qquad \lambda \in l_0, \\
\alpha \in \delta_1 \qquad \mbox{if} \qquad \im \; \lambda (\alpha, \vec{v}_g) =0  \qquad \mbox{for} \qquad \lambda \in l_1.
\end{aligned}\end{equation}
Let also
\begin{equation}\label{eq:}\begin{split}
\mathcal{T}_1 &= \{ \tau^\pm_{0,\alpha}(\lambda), \quad \alpha\in \delta_0, \; \lambda \in l_0\} \cup
\{ \tau^\pm_{1,\alpha}(\lambda), \quad \alpha\in \delta_1, \; \lambda \in l_1\} ; \\
\mathcal{T}_2 &= \{ \rho^\pm_{0,\alpha}(\lambda), \quad \alpha\in \delta_0, \; \lambda \in l_0\} \cup
\{ \rho^\pm_{1,\alpha}(\lambda), \quad \alpha\in \delta_1, \; \lambda \in l_1\} .
\end{split}\end{equation}
Then: A) $\mathcal{T}_1 $ (resp. $\mathcal{T}_2$) provide a minimal set of scattering data which allow one to recover:
\begin{enumerate}
  \item the Gauss factors $S_\nu^\pm (\lambda)$ for $\nu =0, 1$ (resp. $T_\nu^\pm (\lambda)$ for $\nu =0, 1$;

  \item the Gauss factors $D_\nu^\pm (\lambda)$ for $\nu =0, 1$ (resp. $T_\nu^\pm (\lambda)$ for $\nu =0, 1$;

  \item the sewing functions $G_\nu(x,t,\lambda)$ and the scattering matrices $T_\nu (\lambda)$ for each ray $l_\nu$, $\nu =0, \dots, 2h_g -1$;

\end{enumerate}

B)  $\mathcal{T}_1 $ (resp. $\mathcal{T}_2$) allows one to construct the regular solution $\xi^\nu (x,t,\lambda)$ of the RHP;

C)  $\mathcal{T}_1 $ (resp. $\mathcal{T}_2$) allows one to recover the potential $Q_x$ of the Lax operator.

\end{theorem}

\begin{proof}
A) Given $\mathcal{T}_1$ (respectively $\mathcal{T}_2$) and using (\ref{eq:Z_n-cons}) we recover the Gauss factors, i.e.,
the group-valued functions $S_\nu^\pm(\lambda,t)$ (respectively $T_\nu^\pm(\lambda,t)$)  for $\nu=0,1$.
Next, using the $\mathbb{Z}_h$-symmetry we construct $S_\nu^\pm(\lambda,t)$ (respectively $T_\nu^\pm(\lambda,t)$)  for $\nu=2,
\dots , 2h_g-1$, see eq.   (\ref{eq:Z_n-cons}).

The next step consists in recovering the functionals of the integrals of motion $d_{\nu, j}^\pm (\lambda)$. The relations
(\ref{eq:dpm}) provide in principle this possibility. Indeed, the sets $\mathcal{T}_1$ and $\mathcal{T}_1$ then
we know the Gauss factors $S_\nu^\pm(\lambda,t)$ and  $T_\nu^\pm(\lambda,t)$ not only in the typical representations,
but also in all fundamental representations. Of course there are serious difficulties in evaluating the right hand sides
of (\ref{eq:dpm}). However there is no doubt that  they are determined uniquely by any of the sets $\mathcal{T}_1$
and $\mathcal{T}_1$. As a result we have determined all Gauss factors in eq. (\ref{eq:58}), i.e. we obtained all sewing
functions $G_\nu(x,t,\lambda)$ and the scattering matrices $T_\nu (t,\lambda)$.

B) Since the operator $\mathcal{L}$ has no discrete eigenvalues, then the corresponding FAS satisfy a regular RHP.
Therefore the sewing functions $G_\nu(x,t,\lambda)$ determine uniquely the solution $\xi_\nu^\pm(\lambda,t)$ of the
RHP.

C) Since the solution $\xi_\nu^\pm(\lambda,t)$ of the RHP is uniquely determined then the corresponding potential of
$\mathcal{L}$ is determined from eq. (\ref{eq:q-CBC}).

\end{proof}

\subsection{Minimal sets of scattering data for the exceptional algebras}

Here we will formulate the results about the spectral data for the Lax operators related to the
exceptional algebras. The calculations were performed by Maple. The allowed us to calculate the
corresponding vectors $\vec{v}_g$ as well as the scalar products $(\vec{v}_g, \alpha_j)$ and
to establish for each of the exceptional algebras the results in Subsection 4.2.
In particular we found that the subsets of roots $\delta_0$ and $\delta_1$ are
given by the sets of `black` and `white` roots, i.e:
\begin{equation}\label{eq:del01}\begin{split}
\delta_0 \equiv B_g, \qquad \delta_1 \equiv W_g .
\end{split}\end{equation}
Below we list, besides the sets of roots $\delta_0$ and $\delta_1$, also Coxeter automorphism as a
linear operator (matrix) acting on the Euclidean space ${\Bbb E}^r$.

\subsubsection{$\mathfrak{g}\simeq E_6^{(1)}$}
Let $\varepsilon_6 = \frac{1}{\sqrt{3}}(- e_6 - e_7 + e_8)$. Then the set of admissible roots \eqref{eq:e6roots} ca be rewritten as:
\begin{equation}
\label{eq:del0E6}
\begin{aligned}
\alpha_0 &= \frac{1}{2}(e_1 + e_2+ e_3 + e_4 + e_5 + \sqrt{3}\varepsilon_6), \quad
\alpha_1 = \frac{1}{2}(e_1 - e_2 - e_3 - e_4 - e_5 + \sqrt{3}\varepsilon_6), \\
\alpha_2 &= e_1 + e_2, \quad
\alpha_3 = e_2 - e_1, \quad
\alpha_4 = e_3 - e_2, \quad
\alpha_5 = e_4 - e_3, \quad
\alpha_6 = e_5 - e_4.
\end{aligned}
\end{equation}
There are two rays \eqref{eq:l-nu}, corresponding tho the Coxeter automorphism of $E_6^{(1)}$. Roots related to the rays $l_0$ and $l_1$:
\begin{equation}
\label{eq:laE6}
\begin{aligned}
 l_0 &= \arg \lambda =0, &\qquad B_g &\equiv \{ \alpha_1, \alpha_4, \alpha_6 \}; \\
 l_1 &= \arg \lambda =\frac{\pi}{12}, &\qquad W_g &\equiv \{ -\alpha_2, -\alpha_3, -\alpha_5 \};
 \end{aligned}
\end{equation}
The Coxeter number of $E_6^{(1)}$ is $h=12$ and the exponents are $1,4, 5, 7, 8, 11$.

The Coxeter automorphism as a linear operator on the dual Euclidean space $\mathbb{E}_6$:
\begin{equation}
\label{eq:Ce6}
\begin{split}
C_{E_6} = \frac{1}{4}
\left(\begin{array}{ccccccc}
 -3 & -1 & -1 & -1 & -1 & \sqrt{3} \\
 -1 & 1 & -3 & 1 & 1 & -\sqrt{3} \\
 1 & -1 & -1 & -1 & 3 &  \sqrt{3} \\
1 & 3 & -1 & -1 & -1 &  \sqrt{3} \\
  1 & -1 & -1 & 3 & -1 & \sqrt{3} \\
 -\sqrt{3} & \sqrt{3} & \sqrt{3} & \sqrt{3} & \sqrt{3} & 1
 \end{array}\right),
\quad C_{E_6}^{12} = \openone_6, \quad \det C_{E_6} = 1.
\end{split}
\end{equation}
Here  $e_j$, $j=1, \dots, 5$ and $\varepsilon_6$ is the basis of $\mathbb{E}_6$.
The characteristic polynomial of $C_{E_6}$ reads:
\begin{equation}
\label{eq:Pe6}
\begin{split}
P_{E_6}= z^6+ z^5  - z^3 +z +1.
\end{split}
\end{equation}

\subsubsection{$\mathfrak{g}\simeq E_7^{(1)}$}

Introducing the notation
\[
\varepsilon_7=-{e_7-e_8\over \sqrt{2}},
\]
the set of admissible roots \eqref{eq:e7roots} can be rewritten as:
\begin{equation}\label{eq:del0}\begin{aligned}
\alpha_0 &= -\sqrt{2}\varepsilon_7 = e_7-e_8, &\quad \alpha_1 &= \frac{1}{2}(e_1 - e_2 -e_3 -e_4 - e_5 -e_6 + \sqrt{2}\varepsilon_7), \\
\alpha_2 &= e_1 + e_2, \qquad \alpha_2 = e_1 + e_2, &\quad \alpha_3 &= e_2 - e_1, \quad \alpha_4 = e_3 - e_2, \\
\alpha_5 &= e_4 - e_3, \qquad \alpha_6 = e_5 - e_4, &\quad \alpha_3 &= e_6 - e_5.
\end{aligned}\end{equation}
Roots related to the rays $l_0$ and $l_1$ are respectively:
\begin{equation}\label{eq:laE7}\begin{aligned}
 l_0 &= \arg \lambda =0, &\qquad B_g &\equiv \{ \alpha_1, \alpha_4, \alpha_6 \}; \\
 l_1 &= \arg \lambda =\frac{\pi}{18}, &\qquad W_g &\equiv \{ -\alpha_2, -\alpha_3, -\alpha_5, -\alpha_7 \};
 \end{aligned}\end{equation}
The Coxeter number for $E_7^{(1)}$ is $h=18$ and the exponents are $1, 5, 7, 9, 11, 13, 17$.
The Coxeter automorphism acts as a linear operator on the space $\mathbb{E}_7$ as:
\begin{equation}\label{eq:Ce7}\begin{split}
C_{E_7} = \frac{1}{4} \left(\begin{array}{ccccccc} -3 & -1 & 1 & 1 & 1 & 1 & \sqrt{2} \\ -1 & 1 & -1 & 3 & -1 & -1 & -\sqrt{2} \\
 -1 & -1 & -3 & -1 & -1 & -1 & -\sqrt{2} \\  -1 & 1 & -1 & -1 & -1 & 3 & -\sqrt{2} \\  -1 & 1 & -1 & -1 & 3 & -1 & -\sqrt{2} \\
-\sqrt{2} & -\sqrt{2} & -\sqrt{2} & -\sqrt{2} & -\sqrt{2} & -\sqrt{2} & 2
 \end{array}\right), \quad C_{E_7}^9 = -\openone_7, \quad \det C_{E_7} = -1.
\end{split}\end{equation}
Here   $e_j$, $j=1, \dots, 6$ and $\varepsilon_7$ forms the basis in the dual Euclidean space ${\Bbb E}_7$.
The characteristic polynomial of $C_{E_7}$ is:
\begin{equation}\label{eq:Pe7}\begin{split}
P_{E_7}= z^7 + z^8 - z^4 - z^3 +z +1
\end{split}\end{equation}


\subsubsection{$\mathfrak{g}\simeq E_8^{(1)}$}
The set of admissible roots is given by \eqref{eq:e8roots}.
Roots related to the rays $l_0$ and $l_1$ are:
\begin{equation}
\label{eq:laE8}
\begin{aligned}
 l_0 &= \arg \lambda =0, &\qquad B_g &\equiv \{ \alpha_1, \alpha_4, \alpha_6, \alpha_8 \}; \\
 l_1 &= \arg \lambda =\frac{\pi}{30}, &\qquad W_g &\equiv \{ -\alpha_2, -\alpha_3, -\alpha_5, -\alpha_7 \};
 \end{aligned}
\end{equation}
The Coxeter number  for $E_8^{(1)}$ is $h=30$ and the exponents are $1, 7, 11, 13, 17, 19, 23, 29$.
The action of the  Coxeter automorphism as a linear operator on the dual Euclidean space  $\mathbb{E}_7$ is given by
\begin{equation}
\label{eq:CE8}
\begin{split}
C_{E_8} = \frac{1}{4}
\left(\begin{array}{cccccccc}
-3 & -1 & 1 & 1 & 1 & 1 & 1 & -1 \\
-1 & 1 & -1 & 3 & -1 & -1 & -1 & 1 \\
-1 & -3 & -1 & -1 & -1 & -1 & -1 & 1 \\
-1 & 1 & -1 & -1 & -1 & 3 & -1 & 1 \\
-1 & 1 & 3 & -1 & -1 & -1 & -1 & 1 \\
-1 & 1 & -1 & -1 & -1 & -1 & 3 & 1 \\
-1 & 1 & -1 & -1 & 3 & -1 & -1 & 1 \\
1 & -1 & 1 & 1 & 1 & 1 & 1 & 3
 \end{array}\right),
\quad C_{E_8}^{15} = -\openone_8, \quad \det C_{E_8} = -1.
\end{split}
\end{equation}
Here  $e_j$, $j=1, \dots, 6$ and $\varepsilon_7$ form up the basis of ${\Bbb E}_7$.
The characteristic polynomial of $C_{E_8}$:
\begin{equation}
\label{eq:PE8}
\begin{split}
P_{E_8}= z^8 + z^7 - z^5 - z^4 - z^3 +z + 1.
\end{split}
\end{equation}

\subsubsection{$\mathfrak{g}\simeq F_4^{(1)}$}

The set of admissible roots is given by \eqref{eq:f4roots}.
Roots related to the rays $l_0$ and $l_1$ are as follows:
\begin{equation}
\label{eq:laF4}
\begin{aligned}
 l_0 &= \arg \lambda =0, &\qquad B_g &\equiv \{ \alpha_1, \alpha_3 \}; \\
 l_1 &= \arg \lambda =\frac{\pi}{12}, &\qquad W_g &\equiv \{ -\alpha_2, -\alpha_4 \};
 \end{aligned}
\end{equation}
The Coxeter number of $F_4^{(1)}$ is $h=12$ and the exponents are $1, 5, 7, 11$.
The Coxeter automorphism acts as a linear operator on the dual space $\mathbb{E}_4$ via:
\begin{equation}
\label{eq:CF4}
\begin{split}
C_{F_4} = \frac{1}{2}
\left(\begin{array}{cccccccc}
1 & -1 & 1 & -1  \\
-1 & -1 & 1 & 1  \\
1 & -1 & -1 & 1 \\
1 & 1 & 1 & 1
 \end{array}\right),
\quad C_{F_4}^6 = -\openone_4, \quad \det C_{F_4} = 1.
\end{split}
\end{equation}
Here  $e_j$, $j=1, \dots, 4$ is the basis of ${\Bbb E}_4$.
The characteristic polynomial of $C_{F_4}$:
\begin{equation}
\label{eq:PF4}
\begin{split}
P_{F_4}= z^4  - z^2 + 1.
\end{split}
\end{equation}

\subsubsection{$\mathfrak{g}\simeq G_2^{(1)}$}

This comes to be  both the simplest and the most peculiar case. The peculiarity is in the fact, that
though the rank of $G_2$ equals 2, its root system is traditionally written as a set of vectors in the
3-dimensional space $\mathbb{E}_3$ which are all orthogonal to the vector $e_1+e_2+e_3$. The Coxeter
automorphism here is the composition of the Weyl reflection $C_{G_2} = S_{\alpha_1} S_{\alpha_2}$.
One can check that $C_{G_2}$ is expressed as the following linear operator in $\mathbb{E}_3$:
\begin{equation}\label{eq:0G2}\begin{split}
C_{G_2} = \frac{1}{3} \left(\begin{array}{ccc} 2 & -1 & 2 \\  2 & 2 & -1 \\  -1 & 2 & 2  \end{array}\right),
\end{split}\end{equation}
which has the following properties:
\begin{equation}\label{eq:1G2}\begin{split}
\qquad C_{G_2}^2 = \left(\begin{array}{ccc} 0 & 0 & 1 \\ 1 & 0 & 0 \\ 0 & 1 & 0 \end{array}\right) , \qquad
C_{G_2}^4 = \left(\begin{array}{ccc} 0 & 1 & 0 \\ 0 & 0 & 1 \\ 1 & 0 & 0 \end{array}\right) , \qquad
C_{G_2}^6  = \openone_3.
\end{split}\end{equation}
The characteristic polynomial is
\begin{equation}\label{eq:2G2}\begin{split}
P_{G_2} = z^3 -2z^2 + 2z -1 = (z-1) (z^2 -z +1).
\end{split}\end{equation}
Another way to approach the problem is to project the root system of $G_2$ onto the 2-dimensional plane
orthogonal to the vector $e_1 + e_2 + e_3$. Indeed, let
\[
\varepsilon_1 = \frac{1}{\sqrt{2}}(e_2 - e_1),\qquad \varepsilon_2 = \frac{1}{\sqrt{2}}(e_3 - e_2).
\]
Then the set of admissible roots \eqref{eq:g2roots} can be rewritten as:
\begin{equation}
\label{eq:del0G2}
\begin{aligned}
\alpha_0 &= \sqrt{2}(\varepsilon_1 + 2\varepsilon_2) = -e_1 -e_2 + 2e_3, \quad
\alpha_1 = \sqrt{2}\varepsilon_1 = e_2 - e_1, \\
\alpha_2 &= \sqrt{2}(\varepsilon_2 - \varepsilon_1) = e_1 - 2e_2 + e_3,
\end{aligned}
\end{equation}
Roots related to the rays $l_0$ and $l_1$ are:
\begin{equation}
\label{eq:laG2}
\begin{aligned}
 l_0 &= \arg \lambda =0, &\qquad B_g &\equiv \{ \alpha_1\}; \\
 l_1 &= \arg \lambda =\frac{\pi}{12}, &\qquad W_g &\equiv \{ -\alpha_2 \};
 \end{aligned}
\end{equation}
The Coxeter number for $G_2^{(1)}$ is $h=6$ and the exponents are $1, 5$.
This Coxeter automorphism induces an action on the dual space $\mathbb{E}_2$ with  a matrix:
\begin{equation}
\label{eq:CG2}
\begin{split}
C'_{G_2} =
\left(\begin{array}{cccccccc}
0 & 1  \\
-1 & 1
 \end{array}\right),
\quad C'_{G_2}{}^3 = -\openone_2, \quad \det C'_{G_2} = 1.
\end{split}
\end{equation}
Here  $\{\varepsilon_1, \varepsilon_2\}$ form up a basis in ${\Bbb E}_2$. Finally, the characteristic polynomial of $C_{G_2}$ reads:
\begin{equation}
\label{eq:PG2}
\begin{split}
P_{G_2}= z^2  - z + 1.
\end{split}
\end{equation}

\section{Conclusions}

We presented here real Hamiltonian forms of affine Toda field theories related  exceptional untwisted complex Kac-Moody algebras. We
established that the special properties of these models allow us to relate
the construction of the RHF to the study of ${\Bbb Z}_h $ symmetries of the
extended Dynkin diagrams (with $h$ being the Coxeter number of ${\frak g}$). The general construction is illustrated by
several examples of such models related to the exceptional Kac-Moody algebras $E_6^{(1)}$, $E_7^{(1)}$, $E_8^{(1)}$,$D_4^{(1)}$ and $G_2^{(1)}$. We used the  ``dihedral realisation'' of the Coxeter automorphism ${\rm Cox}\, ({\frak g})=S_1\circ S_2$, where $S_1, S_2$ are ${\Bbb Z}_2$-automorphisms of ${\frak g}$. One of these automorphisms extracts the RHF and the other one acts as additional ${\Bbb Z_2}$-reduction on it.

Each of the real Hamiltonian forms obtained above has its dual one.
Indeed, we could define the involution $\mathcal{C}$ using the Weyl reflection $S_2$. The consideration is quite analogous:
simply we interchange the black roots with the white ones. The reformulation of all the above results is rather
obvious. Note that the dual real Hamiltonian forms of 2DTFT {\it are not} equivalent to the ones obtained above.
Indeed, it is easy to check that in all cases treated above the number of the white roots is different from
the number of black roots.

Some additional problems are natural extensions to the results presented
here:

\begin{itemize}
  \item The complete classification of all nonequivalent RHF of ATFT.

\item The description of the hierarchy of Hamiltonian structures of RHF of ATFT (for a review of the infinite-dimensional cases see e.g.
\cite{DriSok,67} and the references therein) It is also an
open problem to construct the RHF for ATFT using some of its
higher Hamiltonian structures.

\item The extension of the dressing Zakharov-Shabat method \cite{ZaSha2}  to the above classes of Lax operators is also an open problem. One of the difficulties is due to the fact that the $\mathbb{Z}_h$ reductions requires dressing factors with $2h$ pole singularities \cite{VSG-88}. This  makes the relevant linear algebraic equations rather involved \cite{Zhu}.

\item Another open problem is to study types of boundary conditions and boundary effects of ATFT's and their RHF \cite{Caudrelier,Doikou1}. This includes types of boundary defects \cite{Doikou2}.

\item The last and technically more involved
problem is to solve the inverse scattering problem for the Lax operator
and thus prove the complete integrability of all these models. The ideas of \cite{AKNS, GVYa} about the interpretation of the inverse scattering method as a generalized Fourier
transform holds true also for the $\mathbb{Z}_h$ reduces Lax operators \cite{SIAM*14, GeYa, Yan2013a, VG-Ya-13,YanVi*2012b}. This may allow one to derive
the action-angle variables for these classes of NLEE.

\end{itemize}

\section*{Acknowledgements}

This work is  supported by the Bulgarian National Science Fund, grant KP-06N42-2.


\begin{thebibliography}{99}

{\small
\bibitem{AKNS}
M. J. Ablowitz, D. J. Kaup, A. C. Newell and H. Segur, {\it The inverse scattering transform -- Fourier analysis for nonlinear   problems},
Stud. Appl.  Math. {\bf 53} (1974), \href{ https://doi.org/10.1002/sapm1974534249}{249--315}.

\bibitem{AMV} M. Adler, P. van Moerbeke and  P. Vanhaecke, {\it Algebraic Integrability, Painlev\'e Geometry and Lie Algebras}, Ergebnisse der Mathematik
und ihrer Grenzgebiete {\bf 47}, Springer-Verlag, Berlin (2004).

\bibitem{BBT} O. Babelon, D. Bernard and M. Talon, {\it Introduction to Classical Integrable Systems}, Cambridge Monographs of Mathematical Physics, Cambridge University Press, Cambridge (2003).

\bibitem{BeCo1} R. Beals and R. R. Coifman, {\it Scattering and inverse scattering for
first order systems}, Comm.  Pure and Appl. Math. {\bf 37} (1984),  \href{ https://doi.org/10.1002/cpa.3160370105}{39--90}.

\bibitem{BeCo2} R. Beals and R. R. Coifman, {\it  Inverse scattering and evolution equations},
Comm.  Pure and Appl. Math.  {\bf 38} (1985), \href{ https://doi.org/10.1002/cpa.3160380103}{29--42}.

\bibitem{Bourb1}
N. Bourbaki, {\it  Lie Groups and Lie Algebras: Chapters 4-6}, Elements of Mathematics {\bf 7},  Springer-Verlag, Berlin-Heidelberg (2002).

\bibitem{BowCor2} H. W. Braden, E. Corrigan, P. E. Dorey and R. Sasaki, {\it Affine Toda field theory and exact S-matrices}, Nucl. Phys. B {\bf 338} (1990), \href{https://doi.org/10.1016/0550-3213(90)90648-W}{689--746}.

\bibitem{BowCor3} H. W. Braden, E. Corrigan, P. E. Dorey and R. Sasaki, {\it Multiple poles and other features of affine Toda field theory}, Nucl. Phys. B {\bf 356} (1991), \href{https://doi.org/10.1016/0550-3213(91)90317-Q}{469--498}.

\bibitem{Cart} R. Carter, {\it Lie Algebras of Finite and Affine Type}, Cambridge Studies in Advanced Mathematics {\bf 96}, Cambridge University Press, Cambridge (2005).

\bibitem{Caudrelier} V. Caudrelier and Q. C. Zhang, {\it Yang-Baxter and reflection maps from vector solitons with a boundary}, Nonlinearity {\bf 27} (2014), \href{https://doi.org/10.1088/0951-7715/27/6/1081}{1081--1103};\\
    J. Avan, V. Caudrelier  and N. Cramp\'e, {\it From Hamiltonian to zero curvature formulation for classical integrable boundary conditions}, J. Phys. A: Math. Theor {\bf 51} (2018), \href{https://doi.org/10.1088/1751-8121/aac976}{30LT01}.

\bibitem{Doikou1} A. Doikou, {\it $A_n^{(1)}$ affine Toda field theories with integrable boundary conditions revisited}, JHEP {\bf 05} (2008), \href{https://doi.org/10.1088/1126-6708/2008/05/091}{091};\\
    J. avan and A. Doikou, {\it Boundary Lax pairs for the $A_n^{(1)}$ Toda field theories}, Nucl. Phys. B {\bf 821} (2009), \href{https://doi.org/10.1016/j.nuclphysb.2009.05.010}{481--505}.

\bibitem{Doikou2} A. Doikou, {\it Jumps and twists in affine Toda field theories}, Nucl. Phys. B {\bf 893} (2015), \href{http://dx.doi.org/10.1016/j.nuclphysb.2015.02.00}{107--121}.

\bibitem{DriSok} V.V. Drinfel'd and  V. V. Sokolov,  {\it Lie algebras and equations of Korteweg-de Vries type}, Journal of Soviet Mathematics, {\bf 30} (1985) \href{http://link.springer.com/article/10.1007/BF02105860}{1975--2036}.

\bibitem{Evans} J. Evans, {\it Complex Toda theories and twisted reality conditions}, Nucl. Phys. B {\bf 390} (1993), \href{https://doi.org/10.1016/0550-3213(93)90393-4}{225--250}.

\bibitem{Evans2} J. Evans  and J. O. Madsen, {\it On the classification of real forms of non-Abelian Toda theories and $W$-algebras}, Nucl. Phys. B \textbf{536}
(1999), \href{https://doi.org/10.1016/S0550-3213(98)00693-2}{657--703}; Erratum-ibid. \textbf{547} (1999) \href{https://doi.org/10.1016/S0550-3213(99)00138-8}{665}.

\bibitem{Evans3} J. Evans  and J. O. Madsen, {\it Real form of non-Abelian Toda theories
and their $W$-algebras}, Phys. Lett. B{\bf 384} (1996), \href{https://doi.org/10.1016/0370-2693(96)00788-5}{131--139}.

\bibitem{FaTa} L.~D.~Faddeev, L.~A.~Takhtadjan, {\it Hamiltonian
Method in the Theory of Solitons}, Springer-Verlag, Berlin (1987).

\bibitem{IP2}  {V. S. Gerdjikov},  {\it Generalised Fourier transforms for  the  soliton
equations. Gauge covariant formulation}, Inverse Problems {\bf 2} (1986), \href{http://iopscience.iop.org/article/10.1088/0266-5611/2/1/005}{51--74}.

\bibitem{VSG-88} V. S. Gerdjikov VS, {\it
    $Z_N$--reductions and  new  integrable  versions  of
    derivative nonlinear Schr\"o\-dinger equations},
    In:  Nonlinear  evolution equations: integrability  and  spectral
    methods, Eds: A. P. Fordy, A. Degasperis and M. Lakshmanan, Manchester
    University Press, Manchester (1981), pp. 367--372.

\bibitem{67}  V. S. Gerdjikov, {\it  Algebraic and analytic aspects of $N $-wave type equations},
Contemporary Mathematics {\bf 301} (2002): \href{http://dx.doi.org/10.1090/conm/301/05158}{35--68}.



\bibitem{VG-13}  V. S. Gerdjikov, {\it Derivative nonlinear Schr\"odinger equations with $\mathbb{Z}_{N}$ and $\mathbb{D}_{N}$ reductions},
Romanian Journal of Physics {\bf 58} (2013), \href{https://rjp.nipne.ro/2013_58_5-6/0573_0582.pdf}{573--582}.

\bibitem{TJ22} V. S. Gerdjikov, {\it  Nonlinear evolution equations related to Kac-Moody algebras $A_r^{(1)}$. Spectral aspects}, Turkish J. Mathematics
(In Press) (2022).

\bibitem{GG}  V. S. Gerdjikov and G. G. Grahovski,  {\it On reductions and real Hamiltonian forms of affine Toda field theories}, J Nonlin. Math. Phys. {\bf 12} (2005), Suppl. 3, \href{https://doi.org/10.2991/jnmp.2005.12.s2.11}{153--168};

V. S. Gerdjikov and G. G. Grahovski, {\it Real Hamiltonian Forms of Affine Toda Models Related to Exceptional Lie Algebras}, SIGMA {\bf 2} (2006), \href{https://doi.org/10.3842/SIGMA.2006.022}{paper 022} (11 pages).

\bibitem{GGMV} V. S. Gerdjikov, G. G. Grahovski, A. V. Mikhailov and  T.I. Valchev, {\it Polynomial bundles and generalised Fourier transforms for integrable equations on {\bf A.III}-type symmetric spaces}, SIGMA {\bf 7} (2011), \href{http://dx.doi.org/10.3842/SIGMA.2011.096}{paper 096} (48 pages);\\
V. S. Gerdjikov, G. G. Grahovski, A. V. Mikhailov and  T.I. Valchev, {\it Rational bundles and recursion operators for integrable equations on {\bf A.III}-type symmetric spaces}, Theor. Math. Phys. {\bf  167} (2011), \href{https://doi.org/10.1007/s11232-011-0058-2}{740--750}.

\bibitem{2} V. S. Gerdjikov, A. V. Kyuldjiev,  G. Marmo  and G. Vilasi, {\it Complexifications and Real Forms of Hamiltonian Structures},
European J. Phys. B {\bf 29} (2002) \href{https://doi.org/10.1140/epjb/e2002-00281-y}{177--182};\\
 V. S. Gerdjikov, A. V. Kyuldjiev,  G. Marmo  and G. Vilasi, {\it Real Hamiltonian forms of Hamiltonian systems}, European Phys. J. B. {\bf 38} (2004) \href{https://doi.org/10.1140/epjb/e2004-00158-1}{635--649}.

\bibitem{GMSV3}  V. S. Gerdjikov,  D. M. Mladenov,  A. A. Stefanov and  S. K.Varbev, {\it
MKdV-type of equations related to $B_{2}^{(1)}$ and $A_{4}^{(2)}$ algebra}, In: Nonlinear Mathematical Physics and Natural Hazards,  B. Aneva and
M. Kouteva-Guentcheva eds., Springer Proc. in Physics {\bf 163} (2014) \href{https://doi.org/10.1007/978-3-319-14328-6_5}{59--69}.

\bibitem{JMP*15} V. S. Gerdjikov,  D. M. Mladenov,  A. A. Stefanov and  S. K.Varbev, {\it
Integrable equations and recursion operators related to the affine Lie algebras $A^{(1)}_{r}$},
 J. Math. Phys. {\bf  56} (2015) \href{https://doi.org/10.1063/1.4919672}{052702}.

\bibitem{JGSP*15} V. S. Gerdjikov,  D. M. Mladenov,  A. A. Stefanov and  S. K.Varbev, {\it
On mKdV equations related to the affine Kac-Moody algebra $A_{5}^{(2)}$},
J. Geom. Symmm. Phys.  {\bf 39} (2015), \href{https://doi.org/10.7546/jgsp-39-2015-17-31}{17--31}.

\bibitem{118a}    V. S. Gerdjikov,  D. M. Mladenov,  A. A. Stefanov and  S. K.Varbev, {\it
MKdV-type of equations related to $\mathfrak{sl}(N, \mathbb{C})$ algebra},
In: Mathematics in Industry, Ed. A. Slavova, Cambridge Scholar Publishing (2015), pp. 335--344.

\bibitem{TMF207}  V. S. Gerdjikov,  D. M. Mladenov,  A. A. Stefanov and  S. K.Varbev, {\it
On the  MKdV type equations related to $A_5^{(1)}$  and $A_5^{(2)}$ Kac-Moody algebras},
Theor. Math. Phys. {\bf 207} (2021), \href{https://doi.org/10.1134/S0040577921050068}{604--625}.

\bibitem{GMSV1}
V. S. Gerdjikov,  D. M. Mladenov,  A. A. Stefanov and  S. K.Varbev, {\it
On an one-parameter family of MKdV equations related to the $\mathfrak{so}(8)$ Lie algebra},
In: Mathematics in Industry, ed. A. Slavova, Cambridge Scholar Publishing (2015), pp. 345--355.

\bibitem{GVYa}  V. S. Gerdjikov,  G. Vilasi and A. B.  Yanovski, {\it
Integrable Hamiltonian hierarchies. Spectral and geometric methods},
Lecture Notes in Physics  {\bf 748}, Springer Verlag, Berlin (2008).

\bibitem{GeYa} V. S. Gerdjikov and A. B.  Yanovski, {\it Completeness of
the eigenfunctions for the Caudrey--Beals--Coifman system},
J. Math. Phys. {\bf 35} (1994), \href{ https://doi.org/10.1063/1.530441}{3687--3725}.

 \bibitem{SIAM*14} V. S. Gerdjikov and A. B.  Yanovski, {\it CBC systems with Mikhailov reductions by Coxeter automorphism
I: Spectral theory of the recursion operators},
Stud. Appl.  Math. {\bf 134} (2015), \href{https://doi.org/10.1111/sapm.12065}{145--180}.

\bibitem{VG-Ya-13}  V. S. Gerdjikov and A. B.  Yanovski, {\it
On soliton equations  with $\mathbb{Z}_{ {h}}$ and $\mathbb{D}_{{h}}$ reductions: conservation laws and generating operators},
J. Geom.  Symmm. Phys. {\bf  31} (2013), \href{https://doi.org/10.7546/jgsp-31-2013-57-92}{57--92}.

\bibitem{GYa*14}  V. S. Gerdjikov   and A. B. Yanovski, {\it Riemann-Hilbert problems, families of commuting
operators and soliton equations}, J. Phys: Conf. Series {\bf 482} (2014) \href{https://doi.org/10.1088/1742-6596/482/1/012017}{012017}.

\bibitem{Grah} G. G. Grahovski and M. Condon, {\it On the Caudrey-Beals-Coifman system and the gauge group action}, J. Nonlin. Math. Phys. {\bf 15} (2008), suppl. 3,  \href{https://doi.org/10.2991/jnmp.2008.15.s3.20}{197--208};\\
    G. G. Grahovski, {\it The Generalised Zakharov-Shabat System and the Gauge Group Action}, J. Math. Phys.  {\bf 53} (2012)  \href{https://doi.org/10.1063/1.4732512}{073512}.



\bibitem{Grah2} G. G. Grahovski, {\it On  reductions and scattering data for the CBC system}, In:
``Geometry, Integrability and Quantization III'', Eds:  I. Mladenov
and G.  Naber, Coral Press, Sofia (2002), pp. \href{https://doi.org/10.7546/giq-3-2002-262-277}{262--277};\\
G. G. Grahovski, {\it On  reductions and scattering data for the generalised Zakharov-Shabat system}, In: ``Nonlinear Physics: Theory and
Experiment. II'', Eds:  M.J.Ablowitz, M.Boiti, F.Pempinelli and
B.Prinari, World Scientific, Singapore, 2003, pp. \href{https://doi.org/10.1142/9789812704467_0010}{71--78}.

\bibitem{Helg}
Helgason S., {\it Differential geometry, Lie groups and Symmetric
Spaces}, Graduate Studies in Mathematics {\bf 34}, AMS, Providence, Rhode Island (2001).

\bibitem{Howlett} R. B. Howlett, L. J. Rylands and D. E. Taylor, {\it Matrix generators for exceptional groups of Lie type}, J. Symb. Comp. {\bf 31} (2000) \href{https://doi.org/10.1006/jsco.2000.0431}{429--445}.


\bibitem{Kac}  V. G. Kac, {\it Infinite-dimensional Lie algebras}, Cambridge University Press, Cambridge (1994).

\bibitem{SasKha}  S. P. Khastgir and R. Sasaki, {\it Instability of solitons in imaginary coupling affine Toda field theory},
 Prog. Theor. Phys. {\bf 95} (1996), \href{https://doi.org/10.1143/PTP.95.485}{485--501};
 S. P. Khastgir and R. Sasaki,, {\it Non-canonical folding of Dynkin diagrams and reduction of affine Toda
theories},  Prog. Theor. Phys. {\bf 95} (1996), \href{https://doi.org/10.1143/PTP.95.503}{503--518}.


\bibitem{Mikh}  A. V. Mikhailov,  {\it The reduction problem and the inverse
scattering problem},   Physica D {\bf 3} (1981) \href{https://doi.org/10.1016/0167-2789(81)90120-2}{73--117}.

\bibitem{OlPerMikh} A. V. Mikhailov, M. A. Olshanetskyand A. M. Perelomov, {\it Two-dimensional generalized Toda lattice}, Commun. Math.
Phys. {\bf 79} (1981) \href{https://doi.org/10.1007/BF01209308}{473--488}.

\bibitem{ZMNP}   S. P. Novikov, S. V. Manakov, L. P. Pitaevsky,  and V. E. Zakharov, {\it Theory of solitons: the inverse scattering method},
 Plenum Press/Consultant Bureau, New York (1984).

\bibitem{Olive}  D. Olive, N. Turok and J. W. R. Underwood, {\it Solitons and the energy-momentum tensor for affine Toda theory},
Nucl. Phys. B {\bf 401} (1993), \href{https://doi.org/10.1016/0550-3213(93)90318-J}{663-697}.

\bibitem{Olove2}  D. Olive, N. Turok and J. W. R. Underwood, {\it Affine Toda solitons and vertex operators}, Nucl. Phys. B {\bf 409} (1993), \href{https://doi.org/10.1016/0550-3213(93)90541-V}{509--546}.

\bibitem{VinOni} A. L.Onishchik and E. B. Vinberg,  {\it Lie groups and algebraic groups},
Springer Series in Soviet Mathematics,  Springer Verlag, Berlin (1990).


 \bibitem{Xu} X. Xu, {\it Kac-Moody Algebras and Their Representations}, Mathematics Monograph Series {\bf 5}, Science Press, Beijing (2006);\\
 X. Xu, {\it Representations of Lie Algebras and Partial Differential Equations}, Springer Nature, Singapore (2017).


\bibitem{Yan2013a}  A. B. Yanovski, {\it
Recursion operators and expansions over adjoint solutions for the Caudrey-Beals-Coifman system with $\mathbb{Z}_p$ reductions of Mikhailov type},
J. Geom. Symm. Phys. {\bf 30} (2013), \href{https:doi.org/10.7546/jgsp-30-2013-105-120}{105--120}.

\bibitem{YanVi*2012b}  A. B. Yanovski and   G. Vilasi, {\it
Geometric theory of the recursion operators for the generalized Zakharov-Shabat system in pole gauge on the algebra $sl(n, \mathbb{C})$: with and without reductions},
SIGMA  {\bf 8} (2012),  \href{http://dx.doi.org/10.3842/SIGMA.2012.087}{paper 087} (23 pages).

\bibitem{ZaSha2} V. E. Zakharov and A. B.  Shabat, {\it
Integration of nonlinear equations of mathematical physics by the method of inverse scattering. II},
Funct. Anal. Appl. {\bf  13} (1979),  \href{https://doi.org/10.1007/BF01077483}{166--174}.


\bibitem{Zamol} A. B. Zamolodchikov, {\it Integrals of motion and $S$-matrix of the (scaled) $T = T_c$ Ising model with magnetic field}, Int. J. Mod. Phys. A {\bf 4} (1989), \href{https://doi.org/10.1142/S0217751X8900176X}{4235--4248}.

\bibitem{Zhu} Z. Zhu and D. G. Caldi, {\it Multi-soliton solutions of affine Toda models}, Nucl. Phys. B {\bf 436} (1995) \href{https://doi.org/10.1016/0550-3213(94)00326-A}{659--678}.

}



\end{thebibliography}
\end{document}